%% file: main.tex
\newif\ifDRAFT 
\DRAFTfalse

\documentclass[11pt]{article}
\usepackage[margin=1in]{geometry}
\usepackage{graphicx} 
\usepackage{mdframed}
\usepackage{amsfonts, amsmath, amssymb, amsthm}
\allowdisplaybreaks
\usepackage{mathrsfs}  
\usepackage{algorithm}
\usepackage[noend]{algpseudocode}
\usepackage{hyperref}
\usepackage{wrapfig}
\hypersetup{colorlinks=true}
\usepackage{color}
\usepackage{array}
\usepackage{cleveref}
\usepackage[inkscapeformat=pdf]{svg}
\usepackage{multirow}
\usepackage{mathtools}
\usepackage{enumitem}
\usepackage{thm-restate}

\usepackage{tabulary} 
\usepackage{booktabs}
\usepackage{arydshln}

\ifDRAFT
    \usepackage{lineno}
    \linenumbers
\fi

\theoremstyle{plain}
\newtheorem{theorem}{Theorem}[section]
\newtheorem{lemma}[theorem]{Lemma}

\newtheorem{observation}[theorem]{Observation}
\newtheorem{claim}[theorem]{Claim}

\theoremstyle{definition}
\newtheorem{definition}[theorem]{Definition}

\crefname{equation}{Eqn.}{Eqns.}

\newcommand{\bydef}{\stackrel{\operatorname{def}}{=}}

\DeclareSymbolFont{bbold}{U}{bbold}{m}{n}
\DeclareSymbolFontAlphabet{\mathbbold}{bbold}

\newcommand{\IGNORE}[1]{}

\newcommand{\pp}{\mathcal{P}}

\newcommand{\prepp}{\widehat{\pp}}

\newcommand{\poly}{\operatorname{poly}}
\newcommand{\polylog}{\operatorname{polylog}}
\newcommand{\dist}{\operatorname{dist}}

\newcommand{\keep}{\mathsf{keep}}
\newcommand{\safe}{\mathsf{safe}}
\newcommand{\postpone}{\mathsf{pstpn}}

\newcommand{\buc}{\mathsf{buc}}
\newcommand{\col}{\mathsf{col}}
\newcommand{\decision}{\mathsf{dcsn}}
\newcommand{\all}{\mathsf{all}}

\newcommand{\sco}{\mathsf{sc}}
\newcommand{\gsco}{\mathsf{gsc}}
\newcommand{\lsco}{\mathsf{lsc}}
\newcommand{\pregsco}{\mathsf{g}\widehat{\mathsf{s}}\mathsf{c}}
\newcommand{\prelsco}{\mathsf{l}\widehat{\mathsf{s}}\mathsf{c}}

\usepackage[normalem]{ulem} 

\title{Parks and Recreation: Color Fault-Tolerant Spanners Made Local}

\ifDRAFT
    \author{Anonymous Authors}
\else
    \author{
    Merav Parter\thanks{Weizmann Institute. Email: \texttt{merav.parter@weizmann.ac.il}. Supported by the European Research Council (ERC) under the European Union’s Horizon 2020 research and
    innovation programme, grant agreement No. 949083.}
    \and
    Asaf Petruschka\thanks{Weizmann Institute. Email: \texttt{asaf.petruschka@weizmann.ac.il}.  Supported by an Azrieli Foundation fellowship, and by the European Research Council (ERC) under the European Union’s Horizon 2020 research and
    innovation programme, grant agreement No. 949083.}
    \and Shay Sapir\thanks{Weizmann Institute. Email: \texttt{shay.sapir@weizmann.ac.il}.  Partially supported by the Israeli Council for Higher Education (CHE) via the Weizmann Data
    Science Research Center.} \and Elad Tzalik\thanks{Weizmann Institute. Email: \texttt{elad.tzalik@weizmann.ac.il}. Supported by the Adams Fellowship Program of the Israel Academy of Sciences and Humanities.} }
\fi
\date{}

\begin{document}

\maketitle

\pagenumbering{gobble}
\input{abstract}
\newpage
\tableofcontents
\newpage

\pagenumbering{arabic}
\input{intro}

\input{techniques}

\input{prelim}

\input{baswana-sen-alternative}
\input{warm-up}

\input{parks-intro}

\input{2k-1-alg}

\input{2k-1-analysis}
\input{VCFT-FINAL-LEVEL}

\phantomsection
\addcontentsline{toc}{section}{References}
\bibliographystyle{alphaurl}
\bibliography{references}

\appendix

\input{VFT-alternative}

\input{Fault-Tolerance-Game}

\input{parks-toolbox}
\input{sampling}

\end{document}

%% file: abstract.tex
\begin{abstract}
We provide new algorithms for constructing spanners of arbitrarily edge- or vertex-colored graphs, that can endure up to $f$ failures of entire color classes.
The failure of even a single color may cause a linear number of individual edge/vertex faults.
This model, related to the notion of hedge connectivity, arises in many practical contexts such as optical telecommunication and multi-layered networks.

In a recent work, Petruschka, Sapir and Tzalik [ITCS `24] gave tight bounds for the (worst-case) size $s$ of such spanners, where 
$s=\Theta(f n^{1+1/k})$ or $s=\Theta(f^{1-1/k} n^{1+1/k})$ for spanners with stretch $(2k-1)$ that are resilient to at most $f$ edge- or vertex-color faults, respectively.
Additionally, they showed an algorithm for computing spanners of size $\tilde{O}(s)$, running in $\tilde{O}(msf)$ sequential time, based on the (FT) greedy spanner algorithm.
The problem of providing faster and/or distributed algorithms was left open therein. 
We address this problem and provide a novel variant of the classical Baswana-Sen algorithm [RSA `07] in the spirit of Parter's algorithm for vertex fault-tolerant spanners [STOC `22].
In a nutshell, our algorithms produce color fault-tolerant spanners of size $\tilde{O}_k (s)$ (hence near-optimal for any fixed $k$), have \emph{optimal locality} $O(k)$ (i.e., take $O(k)$ rounds in the LOCAL model), can be implemented in $O_k (f^{k-1})$ rounds in CONGEST, and take $\tilde{O}_k (m + sf^{k-1})$ sequential time.

In order to handle the considerably more difficult setting of color faults, our approach differs from [BS07, Par22] by taking a novel \emph{edge-centric} perspective, instead of (FT)-clustering of vertices; in fact, we demonstrate that this point of view simplifies their algorithms.
Another key technical contribution is in constructing and using collections of short paths that are ``colorful at all scales'', which we call ``parks''. 
These are intimately connected with the notion of spread set-systems that found use in recent breakthroughs regarding the famous Sunflower Conjecture.
We believe these ideas could potentially find other applications in fault-tolerant settings and spanner-related problems.

\end{abstract}

%% file: intro.tex
\section{Introduction}

Graph spanners, introduced by \cite{PelegS:89}, are sparse subgraphs that preserve the shortest path metric, up to a small multiplicative stretch. 
Formally, a $t$-spanner of an (undirected, weighted) graph $G=(V, E)$ is a subgraph $H\subseteq G$ such that $\dist_H(v,u)\leq t\cdot \dist_G(v,u)$ for every $v,u\in V$.
Spanners are fundamental graph structures that have found a wide-range of applications, especially in distributed settings, e.g., for routing \cite{PelegU:89-routing} and synchronizers \cite{awerbuch1990network}.

In real-world scenarios, spanners may be employed in systems whose components are susceptible to occasional breakdowns.
It is therefore desirable to have spanners possessing resilience to such failures, leading to the notion of \emph{fault-tolerant (FT) spanners}.
These structures were originally introduced in the context of geometric graphs by Levcopoulos et al.~\cite{levcopoulos1998efficient} and Czumaj and Zhao~\cite{czumaj2004fault}, and later on, for general graphs by Chechik et al.~\cite{ChechikLPR09,ChechikLPRJournal}.
In this paper, we focus on the model of \emph{color} fault-tolerance (CFT), recently introduced for spanners by Petruschka, Sapir and Tzalik \cite{PST24spanners},
where the edges (or vertices) of the graph $G$ are given (arbitrarily chosen) colors, and the spanner is required to support any failure of $\leq f$ color classes.

\begin{definition}[$f$-CFT $t$-Spanners]
    An $f$-ECFT (VCFT) $t$-spanner of an edge-colored (or vertex-colored) graph $G$ is a subgraph $H$ such that for every set $F$ of at most $f$ colors in $G$, it holds that $H-F$ is a $t$-spanner of $G-F$.
    (By subtracting $F$ from a graph, we mean deleting all edges/vertices with colors from $F$.)
    %
\end{definition}
Most of the prior work on FT spanners focused on the classical $f$-EFT (VFT) models.
These are defined exactly as above, only with $F$ being a set of at most $f$ edges (vertices).
In fact, these are special cases of CFT spanners, obtained if all edges/vertices have different colors.

\subsection{On the Color Faults Model}
As mentioned above, the conventional E/VFT models assume a given upper bound $f$ on the number of \emph{individual} edge/vertex faults.
This assumption goes back to the fundamental notions of cuts and connectivity, where the resilience of the graph is measured by the smallest number of edges (or vertices) whose removal disconnect the graph. The simplicity of this definition leads to clean structural characterizations of minimum cuts and connectivity, such as Menger's theorem, Nash-Williams tree-packing~\cite{nash1964decomposition}, Gomory–Hu trees~\cite{gomory1961multi}, Karger's cut sampling \cite{Karger00}, among many other foundational graph-theoretic and algorithmic results.

However, as has been observed in many prior works~\cite{PelcP05,farago2006graph,GhaffariKP17,BHPSODA24}, 
these classical faulty models are limited by assuming uncorrelated faulty events,
and the quality of solutions deteriorates significantly with increasing $f$.
One may argue that in realistic settings, such as
modern communication, optical and transportation networks, faulty events usually admit some correlation and structure.
Ideally, one could leverage these properties to efficiently handle even a very large number, say $\Omega(n)$, of individual faults.
In the context of FT spanners, two very recent works have done so in two different such models: the faulty-degree model \cite{BHPSODA24}, and the color faults model \cite{PST24spanners}; the latter is the focus of the current paper.

Color faults represent one of the simplest forms of correlation: 
edges or vertices of the same color fail (or survive) together.
In practical contexts, they have been used to model Shared Risk Resource Groups (SRRG), which capture network survivability issues where a 
failure of a common resource, used by many nodes/links, 
causes all of them to crash
\cite{coudert2007shared, Kuipers12, ZPT11}.
The extension of
classical problems (such as connectivity, minimum cuts and disjoint paths) to the colored setting have been studied over the years under various different terminologies, and are often much harder. 
A notable such problem is that of computing the \emph{hedge connectivity} \cite{GhaffariKP17}, which is the smallest number of colors whose removal disconnects the graph. 
This problem has been also studied under the name of \emph{labeled global cut} \cite{Zhang14,ZhangF16,ZhangFT18,ZhangT20}. A recent work of \cite{JaffkeLMPS23} established the (conditional) quasi-polynomial complexity of the problem, which implies that the upper bound of \cite{GhaffariKP17} is essentially nearly tight. 
 
In contrast, relatively little is known about succinct graph structures (such as spanners, distance oracles, and labeling/routing schemes) designed to tolerate up to $f$ color faults;
such structures were studied recently by \cite{PST24spanners,PSTCLabels24}.
As the coloring is arbitrary, the failure of a single color may cause $\Omega(n)$ edges/vertices to crash, rendering existing E/VFT solutions unsuitable.

While CFT structures are interesting in their own right, they may have applications in seemingly ``non-colored'' settings. 
Such a phenomenon is demonstrated in the recent work of Parter, Petruschka and Pettie~\cite{ParterPP24} on labeling schemes for connectivity under \emph{vertex} faults.
Their approach augmented the given graph with many virtual edges, corresponding to paths through components in some graph decomposition.
Each virtual edge gets a ``color'' according to the component that it represents --- when the latter suffers a (vertex) fault, the corresponding color class of virtual edges becomes faulty (as the paths represented by them could be damaged).  
As it turns out, sparsifying this virtual colored-graph is needed to reduce the label size; the appropriate sparse structure is a CFT  \emph{connectivity certificate}, which is a closely related structure to CFT spanners (see~\cite{PST24spanners}).

\subsection{The Quest for Optimal-Size FT Spanners}
The \emph{existential} size bounds of spanners and FT spanners are relatively well understood by now.
Alth\"ofer et al.~\cite{AlthoferDDJS:93}
proved that for any integer $k \geq 1$, every $n$-vertex graph has a $(2k-1)$-spanner with $O(n^{1+1/k})$ edges, by introducing and analyzing the seminal greedy spanner algorithm.
This tradeoff is believed to be tight by Erd\H{o}s' Girth Conjecture~\cite{erdHos1964extremal}; essentially all lower bounds on (FT) spanner sizes are conditional it.
Turning to FT spanners, the most successful approach in providing tight size bounds has been analyzing FT variations of the greedy spanner by the \emph{blocking-set} technique.
This method was introduced by Bodwin and Patel \cite{BP19}
who settled the optimal size of $f$-VFT $(2k-1)$-spanners to $O(f^{1-1/k}n^{1+1/k})$, matching a lower bound of \cite{BDPW18}.  
The latter also provided a size lower bound of $\Omega(f^{1/2-1/2k}n^{1+1/k})$ for EFT spanners, which was nearly matched by
Bodwin, Dinitz and Robelle \cite{BodwinDR22}.
As mentioned before, recent work \cite{PST24spanners,BHPSODA24} introduced new notions of FT spanners for more realistic (and stronger) faulty models. 
Surprisingly, even though FT spanners in these models can tolerate up to a linear number of individual faults, their size does not increase much (or at all) compared to the spanner size in the standard faulty model.
Yet again, their analysis uses the greedy algorithm and the blocking-set method.
In the CFT model, which is of most interest to us, \cite{PST24spanners} settled the optimal size bounds as $\Theta(f n^{1+1/k})$ for ECFT, and $\Theta(f^{1-1/k} n^{1+1/k})$ for VCFT.

Turning to \emph{computational} aspects, the vast majority of polynomial-time algorithms that produce FT spanners of near-optimal size~\cite{DR20,BodwinDR21,BodwinDR22,BHPSODA24,PST24spanners} are based on modifications of the naive FT greedy algorithm (whose running time is exponential in $f$).
However, greedy-based algorithms have several drawbacks.
First, their running time is often a rather large polynomial, usually lacking natural approaches for major improvements.
Second, it is rather non-versitle; to quote~\cite{BodwinDR21}, ``the greedy algorithm is typically difficult to parallelize or to implement efficiently
distributedly (particularly in the presence of congestion)''.
For this reason, some distributed algorithms have diverged from the greedy approach, while settling on non-optimal sparsity of the spanner (e.g., the CONGEST algorithm of \cite[Section 5.2]{DR20}, based on \cite{DinitzK11}.)

From the abundance of works on FT spanners with near-optimal size, one uniquely stands out as entirely different from the greedy approach: The algorithm for VFT spanners by Parter~\cite{Parter22}.
This algorithm is a fault-tolerant adaptation of the classical Baswana-Sen algorithm~\cite{BaswanaS07},
which is arguably the most versatile spanner algorithm, applicable in numerous computational models~\cite{BaswanaS08,GhaffariK18,BiswasDGMN21,BodwinK16,BernsteinFH21,AhnGM12,KapralovW14}.
As a result, its adaptation to the VFT setting enjoys many remarkable properties:
It can be executed within near-linear sequential time, $O(k)$ rounds (independent of $n$) in the LOCAL distributed model \cite{Linial92}, and in $\polylog(n)$ rounds in the CONGEST model \cite{Peleg:2000}.
Each such property, by itself, has not been achieved by any other algorithm.

\subsection{Our Contribution}

We introduce a Baswana-Sen~\cite{BaswanaS07} style algorithm for faster and/or distributed computation of CFT spanners with near optimal size for any fixed stretch, addressing a problem raised by~\cite{PST24spanners}.
Together with Parter's algorithm for VFT spanners~\cite{Parter22}, this demonstrates the utility and benefits of such approaches over greedy constructions (which enjoy other great advantages, such as many ``existential optimallity'' properties and simplicity).
Adapting Baswana-Sen to the CFT setting turns out to be a rather challenging task.
Extending the clustering approach of \cite{Parter22} to the CFT setting does not work due to inherent barriers, as elaborated in the following \Cref{sect:techniques_and_challenges}.
We also note that the recent expander-based approach of \cite{BHPSODA24} cannot be extended to our setting, as expanders are not resilient to even a single color fault.
To tackle it, we introduce new tools and techniques that
may be useful also for other fault-tolerance or spanner related goals. For example, we show in \Cref{sect:VFT-alternative} how they can simplify Parter's algorithm.
Quantitatively, our results are summarized in the following theorem.

\begin{theorem}\label{thm:key_results}
    Let $G$ be a weighted edge-colored or vertex-colored graph with $n$ vertices and $m$ edges.
    Denote by $s$ the (conditionally) optimal worst-case bound on the size of an $f$-CFT $(2k-1)$-spanner for an $n$-vertex graph, namely:
    \begin{itemize}
        \item $s = \Theta (f n^{1+1/k})$ in the ECFT setting,
        \item $s = \Theta (f^{1-1/k} n^{1+1/k})$ in the VCFT setting.
    \end{itemize}
    There is a randomized algorithm that, with high probability, constructs an $f$-CFT $(2k-1)$-spanner $H$ of $G$, with $|E(H)| \leq s \cdot \log n \cdot 2^{O(k^2)} = \tilde{O}_k (s) $, which runs in:%
    \footnote{The notation $O_k(\cdot)$ hides factors depending only in $k$, and $\tilde{O}$ hides $\polylog(n)$ factors.}
    \begin{itemize}
        \item $O(k)$ LOCAL rounds,
        \item $O_k (f^{k-1})$ CONGEST rounds,
        \item $\tilde{O}_k (m + f^{k-1} s)$ sequential time.
    \end{itemize}
\end{theorem}

Our results are most meaningful in the \emph{constant stretch} regime, where $k=O(1)$.
In this case, even when $f$ is moderately large (say $f \approx n^{1/10k}$), the $f^{k-1}=\poly(f)$ factors only mildly hinder the algorithm's efficiency.
We note that while constant stretch is a central regime for spanners, other regimes like $k = \Theta(\log n)$ (for which standard spanners have linear sparsity) are also of interest, and we leave them open.

Theorem~\ref{thm:key_results} improves on prior work in all three computational models.
In the LOCAL model, one can achieve an $O(\log n)$-rounds algorithm by the approach of~\cite{DR20} using network decompositions, where  clusters of hop-diameter $O(\log n)$ are collected into leader vertices
that can internally run even the naive FT greedy algorithm.%
\footnote{Obtaining clusters with hop-diameter of $O(k)$ in the network decomposition approach of \cite{DR20} will essentially lead to a multiplicative increase of $n^{1/k}$ in the spanner size.}
Our algorithms are truly local in the sense that their locality parameter is $O(k)$, hence independent of the number of nodes in the graph. 
That is, each vertex can determine its edges in the spanner by inspecting only its $O(k)$-radius neighborhood. 
This strong notion of locality has been studied since the seminal work of Naor and Stockmeyer~\cite{NaorS95}. 
Locality $O(k)$ is known to be tight even for standard $(2k-1)$-spanners~\cite{DerbelGPV08}.
In the CONGEST model, one could apply the approach of~\cite{DinitzK11,DR20} and combine it with hit-miss hash families of \cite{KarthikP21}
to obtain algorithms for CFT spanners that run in $\poly(f,\log n)$ rounds, but this would yield spanners of sub-optimal sparsity (as it does in the VFT setting).
The sequential running time improves upon~\cite{PST24spanners} for $k=2,3$, and for any other fixed $k$ when $f \leq m^{1/(k-2)}$.

\paragraph{Discussion and Open Problems.}

The two most intriguing open problems that remain from this work are removing the $f^{k-1}$ factor from the (sequential and CONGEST) time complexities, and extending our results to 
super-constant $k$ (e.g., $k \approx \log n$).
In \Cref{sec:fault_tolerant_game}, we introduce an ``online fault-tolerance'' game, where an $f^k$ factor is necessary.
It illustrates the source of the $f^{k-1}$ factor in Theorem~\ref{thm:key_results}, and suggests that certain relaxations of our ``online'' approach are necessary to remove this dependence.
Nevertheless, we conjecture
that for any fixed stretch, there is an algorithm whose running time is $\tilde{O}(m)$ and outputs a spanner of near-optimal sparsity. We also believe that there is a CONGEST algorithm that runs in $\tilde{O}(1)$ 
rounds with similar size guarantees. 
Finally, the work of \cite{Parter22} provided a PRAM (CRCW) algorithm with $\tilde{O}(m)$ work and $\tilde{O}(1)$ depth for computing near-optimal VFT spanners (for unweighted graphs).
Finding efficient parallel algorithms for near-optimal CFT spanners seems to pose a major challenge, left untreated in the current work.
Another tantalizing direction is designing Baswana-Sen-like algorithms in other fault models, e.g., the faulty-degree model of~\cite{BHPSODA24}.

%% file: techniques.tex
\subsection{Technical Highlights}\label{sect:techniques_and_challenges}

In this section, we discuss the inherent barriers in extending Baswana-Sen~\cite{BaswanaS07} and Parter's algorithms~\cite{Parter22} to the color faults setting, and 
give a birds-eye overview of how our approach mitigates these key challenges.
The warm-up provided in \Cref{sec:warm-up} is intended to demonstrate our central novelties in more detail.

\paragraph{From Vertex Clusters to an Edge-Centric Perspective.}
Both algorithms of \cite{BaswanaS07,Parter22} rely on clustering vertices into low-depth trees, whose roots are called \emph{centers}, and
the clusters evolve in $k$ levels, $i=0,1,\dots,k-1$.
Initially, all vertices are $0$-clustered (as $0$-centers).
In each level, some of the $i$-clusters dissolve, and their vertices either join surviving clusters and become $(i+1)$-clustered, or otherwise, they are completely taken care of, meaning all their incident edges already have a good stretch guarantee (even under faults) in the current spanner.

The colored setting poses a challenge for this clustering approach, as exemplified by considering the first clustering step in \cite{BaswanaS07} and \cite{Parter22}.
Both are based on a hitting-set argument:
sampling every vertex to be a $1$-center with some probability $p$ (which is $n^{-1/k}$ in \cite{BaswanaS07} and $(n/f)^{-1/k}$ in \cite{Parter22}) gives that (w.h.p) each high-degree vertex sees many sampled centers among its neighbors.
Every vertex that sees enough centers,
joins their clusters and becomes $1$-clustered.
Low-degree vertices may not be clustered, 
but their not-too-many incident edges can be added to the spanner.
While in \cite{BaswanaS07}, it is enough for a vertex to join one cluster to be considered $1$-clustered, in the adaptation of \cite{Parter22} to the VFT setting, it should join $\Theta(f)$ clusters, which guarantees a surviving edge to at least one neighboring center even after the failure of $\leq f$ vertices.
To get a similar guarantee in the (edge) color faults setting, we would like $1$-clustered vertices to have $\Theta(f)$ edges \emph{of different colors} going to sampled centers.
The hitting set argument works well for vertices with high \emph{color-degree}, namely, those that have many incident edges of distinct colors.
However, one cannot handle the remaining low color-degree vertices by simply adding all their incident edges to the spanner, as their actual degree might be very large.

Our approach moves away from the clustering perspective, and takes an \emph{edge-centric} point of view.
The idea is to gradually \emph{decide} on each edge whether to include it in the spanner or discard it from consideration.
Edges that are still undecided before executing level $i$ are guaranteed to have certain collections of short paths from their endpoints to $i$-centers that are already included in the spanner, and are fault-tolerant \emph{from the perspective of the edge}.
We demonstrate this for level $0$: the endpoints of edge $e$ with color $c(e)$ that remains undecided for level $1$ are either provided with $\Theta(f)$ colorful spanner-edges to $1$-centers, or with a single spanner-edge \emph{of the same color $c(e)$} to a $1$-center; in both cases, for any set of at most $f$ failing colors $F$ with $c(e) \notin F$, there is a surviving spanner-edge to a $1$-center.
We show that such a guarantee can be achieved  for the undecided edges 
without adding too many spanner-edges in level $0$.

Interestingly, both algorithms of ~\cite{BaswanaS07} and \cite{Parter22} can be viewed in this edge-centeric perspective.
In \Cref{sect:baswana-sen}, we describe the Baswana-Sen algorithm from this point of view, which serves as a preliminary introduction to our approach.
Additionally, we showcase its utility by providing a simplified, $1.5$-page description and analysis of Parter's algorithm, found in \Cref{sect:VFT-alternative}.

\paragraph{From (Vertex) Disjointness to (Color) Spreadness.}
Let us now delve into the approach taken by Parter's algorithm to handle vertex faults.
The key structural lemma lying in the heart of the stretch analysis~\cite[Lemma 3.12]{Parter22}, can be roughly described as follows.
An edge $e= \{u,v\}$ is discarded from consideration in level $i$ only when the following situation occurs.
The current spanner contains (1) a collection $\pp_u$ of $\Theta(kf)$ paths of length $i$ stemming from $u$, that are vertex-disjoint (except for sharing $u$), and (2) a collection $\prepp_v$ of paths of length $i+1$ stemming from $v$, that are vertex-disjoint (except for sharing $v$),
with the following property:
Each $P \in \pp_u$ intersects some path $Q \in \prepp_v$ in a vertex (different from $u$).
A trivial yet crucial observation is that concatenating (the prefixes of) $P$ and $Q$ at their intersection point gives a short (length $\leq 2i+1$) $u$-$v$ path in the current spanner.
Using the vertex-disjointness properties of $\pp_u$ and of $\prepp_v$, and the fact that $\pp_u$ is large enough, the lemma shows how to construct at least $f+1$ such concatenated $u$-$v$ paths that are internally-vertex-disjoint; this proves that no matter which $f$ vertex faults occur, the stretch of $e$ in the spanner is at most $2i+1 \leq 2k-1$, so it can be safely discarded.

Unfortunately, this reasoning totally breaks if one tries to replace vertex-disjointness with \emph{color-disjointness} (aiming to eventually get $f+1$ color-disjoint $u$-$v$ paths); this is because two paths whose colors intersect may not share any vertices, and hence cannot be concatenated.
Put differently, vertex faults (or, for that matter, also edge faults) have a certain \emph{topological} structure which can be utilized: all paths that are destroyed when a vertex fails are intersecting (at that vertex).
In contrast, the paths destroyed by one failing color do not have any topological structure; we only get that their \emph{color-sets} are intersecting. 
Roughly speaking, this suggests that CFT spanners is a mixture of graph problem and a set system problem.

Instead of demanding that path collections are color-disjoint, the heart of our approach is based on ensuring a \emph{spreadness} property of their color-sets; a path collection satisfying it is called a \emph{park}.
Informally, a collection of length-$i$ paths is a park if for any set of colors $J$, the number of paths whose color-set contains $J$ is upper-bounded by roughly $f^{i-|J|}$.
We say that the park is \emph{$J$-full} if this bound is saturated.
The formal definitions regarding parks appear in \Cref{sec:parks}.
Their key fault-tolerant property, formally stated in~\Cref{lemma:union_bound_scores}, is the following: If a park is $J$-full, and $f$ colors \emph{not in $J$} fail, then the park contains a surviving path. 
Hence, this notion of parks is natural in the study of color fault-tolerance. 
Interestingly, similar spreadness properties have been used to achieve recent breakthroughs in combinatorics related to \emph{sunflowers} in set-systems~\cite{LovettSpreadSetSystems20}; 
\Cref{sec:parks_sunflowers}
discusses conceptual connections between fault-tolerance, parks/spreadness and sunflowers.

To handle the lack of topological structure, in addition to having ``global'' parks of paths from a vertex to the set of centers, we also maintain ``local'' parks of paths between a vertex and \emph{one specific center}.
Our stretch argument combines both types of parks.
To prove we can safely discard $e=\{u,v\}$, we first apply the fault-tolerant property in a global park stemming from $u$ to obtain a non-faulty path to some center $s$.
Then, hinging on having non-faulty colors in this path, we apply the fault-tolerant property again, but now in a local park between $v$ and $s$. This gives another non-faulty path that can be concatenated with the previous to yield a short surviving $u$-$v$ path. 
Such a detailed argument is demonstrated in \Cref{sec:warm-up}.

%% file: prelim.tex
\section{Preliminaries}

\subsection{Notations and Terminology}
Throughout, we denote the undirected, weighted and colored input multi-graph by $G = (V,E,w,c)$, where $|V| = n$, $|E| = m$, $w: E \to \mathbb{R}_+$  is the weight function, and the coloring is $c: E \to C$ in the ECFT setting, or $c: V \to C$ in the VCFT setting
($C$ is the set of possible colors).
The fault parameter is a positive integer denoted by $f$. The stretch parameter is $(2k-1)$, where $k$ is a positive integer.
It is instructive for the reader to focus on the case where $k$ is a fixed constant.
Formally, in the ECFT setting we require that $k \leq \varepsilon \sqrt[3]{\log n}$ for a sufficiently small constant $\varepsilon>0$, and in the VCFT setting we require $k \leq \varepsilon \sqrt[3]{\log(n/f)}$.
We say that a color set $F \subseteq C$ \emph{damages} an edge%
\footnote{
    We slightly abuse notation and write $e = \{v,u\}$ to say that $u,v \in V$ are the endpoints of $e$, even though there might be several parallel edges with the same endpoints.
}
$e = \{v,u\} \in E$ if its failure causes $e$ to fail, i.e., if $c(e) \in F$ in the ECFT setting, or if $\{c(v),c(u)\} \cap F \neq \emptyset$ in the VCFT setting.
When $G'$ is a subgraph of $G$, then $G'-F$ denotes the subgraph of $G'$ obtained by deleting all edges damaged by $F$.
We mainly focus on the ECFT setting, unless explicitly stated otherwise. 

We now introduce terminology and notations regarding paths and path collections.
We say that a path collection is \emph{stemming from $v \in V$} if $v$ is the first vertex of every path in the collection. We usually denote such a path collection by $\pp_v$.
We say that a path collection is \emph{ending at $S \subseteq V$} if every path in the collection ends in a vertex from $S$.
For every vertex $s$, we denote by $\pp_{v,s}$ the set of paths in $\pp_v$ whose last vertex is $s$.
When $e$ is an edge with endpoint $v$, and $P$ is a path starting at $v$, we denote by $e \circ P$ the path obtained by concatenation of $e$ to $P$ through $v$. When $\pp_v$ is a path collection stemming from $v$, we denote $e \circ \pp_v = \{e \circ P \mid P \in \pp_v\}$. 
When $P$ is a path, $c(P)$ denotes the set of colors appearing on $P$.
Given a path collection $\pp$ and a color set $J$, we will often be interested in the sub-collection of paths $P \in \pp$ such that all colors in $J$ appear on $P$. We call such a sub-collection a \emph{link}.%
\footnote{In most literature the collection below is called the \emph{star} of $J$. Since the term ``star'' has other common usages in graph algorithms literature, we slightly abuse the standard naming and call the described collection \emph{link}.}

\begin{definition}[Link]\label{def:link}
    When $\mathcal{P}$ is a path collection, and $J \subseteq C$ is a subset of colors, the $J$-link of $\pp$ is defined as
    $
    \pp[J] := \{ P \in \pp \mid  J \subseteq c(P)\}.
    $
\end{definition}

%% file: baswana-sen-alternative.tex
\subsection{An Alternative View of the Baswana-Sen Algorithm}\label{sect:baswana-sen}
To demonstrate our edge-centric perspective, we provide an alternative description of the Baswana-Sen algorithm.
The algorithm gradually constructs the spanner $H$ by adding edges.
It works in $k$ \emph{levels}, numbered $i=0,1,\dots,k-1$. 
Denote by $H_i$ the subgraph of $H$ consisting of all edges that were added in levels $0,1,\dots,i-1$, so the output is $H = H_k$.
Level $i$ has a corresponding set of $i$-\emph{centers} $S_i$, where $S_0 = V$, and $S_i$ is formed by sampling each $s \in S_{i-1}$ independently with probability $p = n^{-1/k}$.

The input for level $i$ is a set $E_i$ of \emph{undecided} edges, where $E_0 = E$.
During level $i$, the algorithm decides for each $e \in E_i$ on one of three options: (i) \emph{keep} $e$ as an edge of $H$, (ii) declare that $e$ is \emph{safe} and can be discarded from consideration, or (iii) \emph{postpone} $e$ to the next level,
by including it in $E_{i+1}$.
The following invariants are maintained:
\begin{itemize}
\item[]
\begin{enumerate}[label={({\bfseries BS\arabic*})}]
    \item If $e = \{v,u\} \in E - E_i$, then $\dist_{H_i} (u,v) \leq (2i-1) w(e)$. \label{inv:BS1}
    
    \item If $e = \{v,u\} \in E_i$, then $H_i$ contains a path $P_u$ from $u$ to some $i$-center, with $i$ edges, each of weight at most $w(e)$. \label{inv:BS2}
\end{enumerate}
\end{itemize}
At initialization ($i=0$), Invariant~\ref{inv:BS1} holds vacuously as $E-E_0 = \emptyset$, and the path $P_u$ of Invariant~\ref{inv:BS2} is just the vertex $u \in S_0$ (with $0$ edges).

During level $i$,
each vertex $v \in V$ processes its incident edges in $E_i$, one by one, in \emph{non-decreasing} order of weight.
It gradually constructs a collection  $\prepp_v$ of $v$-to-$S_i$ paths, supported on $H$, according to the following rules:
\emph{(``local'')} no two paths go to the same $i$-center,
and
\emph{(``global'')} there are at most $O(1/p \cdot \log n)$ paths in total.
When edge $e = \{v,u\} \in E_i$ is processed, $v$ considers the path $e \circ P_u$ (the concatenation of $e$ with its $P_u$ of \ref{inv:BS2} for level $i$), and adds it to $\prepp_v$ only in case this does not violate any of the rules.
Then, $v$ decides on $e$ as follows:
\begin{itemize}
    \item \emph{(``Keep'')}
    If $e \circ P_u$ is added, $v$ decides to \emph{keep} $e$ in $H$.\\
    Note that the global rule guarantees that $v$ only keeps $O(1/p \cdot \log n) = O(n^{1/k} \log n)$ edges.

    \item \emph{(``Safe'')}
    If adding $e \circ P_u$ would violate the local rule, $v$ decides that $e$ is \emph{safe} to be discarded.\\
    Indeed, this means there was already some kept edge $e' = \{v,u'\} \in E_i$ with $w(e') \leq w(e)$, such that $e' \circ P_{u'}$ and $P_u$ end in the same $i$-center.
    Their concatenation gives a path of $2i+1$ edges, each with weight at most $w(e)$, that connects $v,u$ in the current $H$.
    This is precisely what is needed to maintain Invariant~\ref{inv:BS1} for the next, $i+1$ level.

    \item \emph{(``Postpone'')}
    If adding $e \circ P_u$ would violate the global rule, $v$ decides to \emph{postpone} $e$.\\
    In this case, there exist $\Omega(1/p \cdot \log n)$ different $i$-centers appearing as endpoints in $\prepp_v$.
    Thus, with high probability, there is a path $P_v \in \prepp_v$ ending at an $(i+1)$-center.
    The algorithm provides the postponed $e$ with this $P_v$ to maintain Invariant~\ref{inv:BS2} for the next, $i+1$ level.%
    \footnote{
        Note the side switching phenomenon: We assumed Invariant~\ref{inv:BS2} (for level $i$) provided $e$ with a path $P_u$ starting at $u$.
        But when $v$ postpones $e$, we find a path $P_v$ \emph{starting at $v$} to maintain Invariant~\ref{inv:BS2}.
        Formally, we insert $e$ into $E_{i+1}$ only if it is postponed by both its endpoints, so Invariant~\ref{inv:BS2} always holds ``from both sides''.
    }
    
\end{itemize}

In the last, $k-1$ level, we have $|S_{k-1}| = O(n p^{k-1} \log n) = O(1/p \cdot \log n)$ with high probability.
Hence, the global rule cannot be violated, so there are no postponed edges and $E_k = \emptyset$.
Thus, by Invariant~\ref{inv:BS1}, $H = H_k$ is indeed a $(2k-1)$-spanner of $G$.
In each level, each vertex only adds $O(n^{1/k} \log n)$ edges to $H$, so it has only  $O(k n^{1+1/k} \log n)$ edges in total.

%% file: warm-up.tex
\section{Warm-Up: $f$-ECFT $3$-Spanner}\label{sec:warm-up}

To warm up, we show a Baswana-Sen-based algorithm for constructing an $f$-ECFT $3$-spanner of $G$ (i.e., with $k=2$) of near-optimal $\tilde{O}(f n^{3/2})$ size.
For the sake of simplicity, we do this under the (mild) assumption that $G$ is a simple graph with no parallel edges.
We start by also assuming a much more restrictive assumption: that the coloring of its edges is \emph{legal}, meaning that every two adjacent edges differ in color.
Then, we end the section by discussing how it is removed.

\paragraph{The First Level ($i=0$).}
Our simplifying assumptions make the execution of level $0$ extremely easy. 
(In fact, under these assumptions, one can use the clustering approach of \cite{BaswanaS07}.)
Each vertex keeps (i.e., inserts to $H$) its $O(f n^{1/2} \log n)$ incident edges (from $E = E_0$) of lowest weight, which keeps us within our size budget for $H$.
There are no ``safe'' decisions. 
Edges that are not kept (by any endpoint) are postponed to $E_1$. 
As in ``vanilla'' Baswana-Sen, the centers $S_1 \subseteq S_0 = V$ are formed by sampling each vertex independently with probability $p = n^{-1/2}$.

Consider a postponed edge $e = \{v,u\} \in E_1$.
Then each of its endpoints, say $u$, must have kept $\Omega(f n^{1/2} \log n)$ of its incident edges.
With high probability, at least $2f$ of them go to sampled centers, i.e., have their non-$u$ endpoint in $S_1$.
Our legal coloring assumption implies that each of these edges has a different color.
In other words, the following property, analogous to Invariant~\ref{inv:BS2} in the ``vanilla'' Baswana-Sen algorithm, holds (w.h.p) before we start (the last) level $1$:
\begin{itemize}
    \item[] If $e = \{u,v\} \in E_1$, then $H_1$ contains a set $\pp_u$ of $2f$ $u$-to-$S_1$ edges (or, paths of length $1$), \emph{with different colors}, all with weight at most $w(e)$.
\end{itemize}
The colorfulness of $\pp_u$, implied by the legal coloring assumption, is crucial: it ensures that one edge in this set will survive, no matter what $f$ color faults occur (the reason we require $2f$ edges, and not just $f+1$, will become clear shortly). 

\paragraph{The Last Level ($i=1$).}
Again,
each vertex $v\in V$ processes its incident edges from $E_1$ in non-decreasing order of weight.
During this process, $v$ gradually constructs a collection $\prepp_v$ of paths stemming from $v$ and ending in $S_1$, supported on $H$.
We maintain the following (``local'') rule for each $s \in S_1$:
\begin{itemize}
\item[]
    For every set of colors $J \subseteq C$,  $\prepp_{v,s}[J]$ contains at most $(2f)^{2-|J|}$ paths.%
    \footnote{Recall that by \Cref{def:link}, $\prepp_{v,s}[J]$ is the restriction of $\prepp_{v,s}$ to paths containing the colors in $J$.}
\end{itemize}
Note that the local rule with $J = \emptyset$ ensures that $\prepp_v$ contains at most $O(f^2 \cdot |S_1|)$ paths in total, which is $O(f^2 \cdot pn \log n) = O(f^2 n^{1/2} \log n)$
with high probability.

To process $e = \{v,u\} \in E_1$, $v$ executes a \emph{voting} mechanism. Each of the $2f$ colorful one-edge paths $P \in \prepp_u$ either votes $\keep$ or votes $\safe$, as follows:
\begin{itemize}
    \item \emph{(``Vote $\keep$'')} If $e \circ P$ could be added to $\prepp_v$ without violating the local rule of $\prepp_{v,s}$, where $s$ is the center at the end of $P$, then $P$ votes $\keep$.
    \item \emph{(``Vote $\safe$'')} Otherwise, $P$ votes  $\safe$.
\end{itemize}
The decision of $v$ regarding $e$ is based on the majority vote.
The vertex $v$ declares $e$ as safe, and discards it, if there are at least $f+1$ $\safe$ votes.
In the complementary case, $v$ decides to keep $e$, i.e., adds it to $H$, and updates $\prepp_v$ by adding to it all the paths $e \circ P$ such that $P \in \pp_u$ voted $\keep$, which does not violate the local rule.%
\footnote{
    Here, we also use the assumption that the graph is simple, so different edges in $\pp_u$ go to different centers $s \in S_1$.
}
This completes the description of the algorithm.

\paragraph{The Size Argument.}
We have already noted that the first level $0$ adds $O(fn^{3/2} \log n)$ edges to $H$.
Consider now the edges kept by some vertex $v$ in the last level.
With each such kept edge, the total number of paths in $\prepp_v$ increased by at least $f$, as there were at least $f$ $\keep$ votes.
But, as we have argued before, (w.h.p) this number cannot exceed $O(f^2 n^{1/2} \log n)$.
Thus, there could only be $O(f n^{1/2} \log n)$ edges that are kept by $v$.
So, overall, only $O(f n^{3/2} \log n)$  edges are added in level $1$.

\paragraph{The Stretch/Safety Argument.}
The only edges from $G$ that are missing in $H$ are those declared safe in the last level.
Consider any such edge $e = \{v,u\} \in E_1$, declared safe by $v$.
To complete the stretch argument,
it suffices to prove that given any set $F$ of at most $f$ faulty colors not damaging $e$ (i.e., $c(e) \notin F$), it holds that $\dist_{H-F} (u,v) \leq 3 w(e)$.
To this end, we focus on the state of the algorithm when $v$ considers $e$.

At least $f+1$ of the (one-edge) paths in $\pp_u$ voted $\safe$, and as their edges are of different colors, there exists a $\safe$-voting $P_1 \in \pp_u$ whose color is not in $F$.
Thus, there must exist a subset $J$ of the colors on $e \circ P_1$ (and hence $J \cap F = \emptyset$) such that $\prepp_{v,s}[J]$ contains exactly $(2f)^{2-|J|}$ paths.
Out of them, the faulty ones are $\bigcup_{c \in F} \prepp_{v,s}[J \cup \{c\}]$.
By the local rule at $\prepp_{v,s}$, for each $c \in F$, $\prepp_{v,s} [J \cup \{c\}]$ has at most $(2f)^{2-|J \cup \{c\}|} = (2f)^{1-|J|}$ paths (as $c \notin J$).
So, in total, there can be no more than $|F| \cdot (2f)^{1-|J|} < (2f)^{2-|J|}$ such faulty paths.
Thus, there is a non-faulty path $P_2 \in \prepp_{v,s}[J]$.

As $P_2$ already appears in $\prepp_v$ before $e$ is considered, there is some kept edge $e' = \{v,u'\} \in E_1$ with $w(e') \leq w(e)$ such that $P_2 = e' \circ P'$ with $P' \in \pp_{u'}$, hence both edges on $P_2$ weigh at most $w(e)$.
The single edge of $P_1$ also weighs at most $w(e)$.
So, by joining $P_1$ and $P_2$ at their common endpoint $s$, we see that $\dist_{H-F} (u,v) \leq 3 w(e)$.

\paragraph{Removing the Legal Coloring Assumption.}
We now sketch how the legal coloring assumption can be removed, starting with the modification to the first level $0$.
Again, each vertex $u$ goes over its incident edges in $E_0 = E$, in non-decreasing order of weight.
Now, $u$ gradually constructs a collection $\prepp_u$ of one-edge $u$-to-$S_0$ paths, subject to the following (``global'') rule:
\begin{itemize}
\item[] 
    For every set of colors $J \subseteq C$, 
    $\prepp_u[J]$ contains at most  $O(1/p \cdot \log n) \cdot (2f)^{1-|J|}$ paths.
\end{itemize}
In other words, this rule says that $\prepp_u$ contains only $O(1/p \cdot \log n)$ edges of any given color, and only $O(f/p \cdot \log n)$ edges overall.
When processing $e = \{u,v\} \in E_0$, $u$ decides to keep $e$ if it can be added to $\prepp_u$ without violating the global rule, and otherwise, it postpones $e$ to $E_1$ (as before, there are no ``safe'' decisions in level $0$).
If $u$ postpones $e$ to $E_1$, there are two possible cases:
\begin{itemize}
    \item Case I: $e$ is postponed because $\prepp_u$ is of size $\Omega(f/p \cdot \log n)$.
    
    Then, we can divide the edges in $\prepp_u$ to $2f$ buckets of size $\Omega(1/p \cdot \log n)$, such that edges of the same color are in the same bucket.
    Hence, with high probability, there will be a $u$-to-$S_1$ edge from every bucket, providing $e$ with a set $\pp_u$ of $2f$ $u$-to-$S_1$ edges of different colors.
    So, in this case we get the same property that we had under the legal coloring assumption.

    \item Case II: $e$ is postponed because there are $\Omega(1/p \cdot \log n)$ edges of color $c(e)$ in $\prepp_u$.

    Then, with high probability, there is a $u$-to-$S_1$ one-edge path $P_u$ of color $c(e)$.
    In this case, we denote $\pp_u = \{P_u\}$.
    Intuitively, from the viewpoint of $e$, this $\pp_u$ is ``as good'' as the colorful $\pp_u$ of Case I: In both cases, if $F$ is a set of failing color \emph{that does not damage $e$}, then there is a surviving path in $\pp_u$.
\end{itemize}

We now discuss the modification for the last level ($i=1$).
The only difference is that we now need to handle ``Case II edges''.
In order to do so, we slightly augment the local rule for each $s \in S_1$:
\begin{itemize}
\item[]
    For every set of colors $J \subseteq C$, $\prepp_{v,s}[J]$ contains at most $(2f)^{2-|J|}$ paths, \emph{out of which at most $(2f)^{1-|J|}$ are monochromatic}.
\end{itemize}
The treatment and analysis of ``Case I'' edges is exactly as before.%
\footnote{If a ``Case I'' edge $e=\{v,u\}$ has a spanner-edge of color $c(e)$ among the $2f$ edges in $\pp_u$, then we treat it as ``Case II'' by ignoring all other edges of $\pp_u$. Thus, ``Case I'' edges cannot violate the monochromatic part of the rule.}
When $v$ considers a ``Case II'' edge $e = \{v,u\} \in E_1$, then $e \circ P_u$ is a monochromatic $c(e)$-colored path.
If adding this path to $\prepp_v$ would violate the rule, then $e$ is declared safe and discarded. 
Otherwise, $e$ is kept and added to $H$, and $e \circ P_u$ is added to $\prepp_v$.
It follows from the rule (with $J=\emptyset$) that the number of ``Case II'' edges kept by $v$ is $O(f|S_1|)$, which is $O(f n^{1/2}\log n)$ with high probability.

For the stretch argument, suppose $e$ is declared safe. Let $F \subseteq C$ be a set of $\leq f$ faulty colors not damaging $e$.
Denote by $s\in S_1$ the other endpoint of $P_u$.
Then, there exists $J\subseteq c(e \circ P_u) = \{c(e)\}$ (and hence $J \cap F = \emptyset$) such that $\prepp_{v,s}[J]$ contains exactly $(2f)^{1-|J|}$ monochromatic paths.
If $J=\{c(e)\}$, then there is a $c(e)$-colored monochromatic path $P_2 \in \prepp_v [J]$, so its concatenation with $P_u$ gives a non-faulty path of length $\leq 3w(e)$ in $H$.
If $J=\emptyset$, then $\prepp_{v,s}[J]$ contains $2f$ monochromatic paths of different colors, so at least one is non-faulty, and we proceed similarly.

%% file: parks-intro.tex
\section{Parks: Color-Spread Path Collections}\label{sec:parks}

We now introduce the key notions of \emph{parks} and their associated score functions, lying at the heart of our algorithm.
These provide a unified approach that generalizes all the different rules of the warm-up $f$-ECFT $3$-spanner algorithm. 
For example, an edge $e$ postponed from level $0$ to level 1 resulted in two options: it was either concatenated to $2f$ colorful spanner-edges, or to one spanner-edge of the same color $c(e)$.
The formal definitions of parks and scores enable us to treat these options as ``equivalent'', without having to go through elaborate case analysis.
They also have the advantage of smoothly transitioning between the ECFT setting and the VCFT setting, allowing us to treat the latter almost identically as the former; only the last level requires additional work.

Informally, a park is a ``spread'' collection $\pp$ of paths, that is, for every subset of colors $J \subseteq C$, there are not too many paths containing $J$ among their colors (i.e., $\pp[J]$ is not too large).
Our measure of ``small/large'' path collections is slightly more subtle than simply the number of paths in the collection, and is based on a notion of \emph{scores} for paths,
which takes into account the resilience of a path to potential color faults. 
We now formally provide a general, parameterized definition of scores and parks. 
For the standard use of these definitions, it suffices that $\beta=1$ and $\alpha=2$, as was in the warm-up (\Cref{sec:warm-up}).

\begin{definition}[Score Function]\label{def:score}
    Let $\alpha>1$, $0<\beta \leq 1$.
    The \emph{$(\alpha,\beta)$-score function} $\sco^{\alpha,\beta}(\cdot)$ assigns a path $P$ with score
    $\sco^{\alpha,\beta} (P) := \beta (\alpha f)^{-|c(P)|}$.
    For $J \subseteq C$, the \emph{induced score function on $J$} is
    \[
    \sco^{\alpha,\beta}_J (P) := 
        \begin{cases}
            \beta (\alpha f)^{-|c(P)-J|} & \text{if $J \subseteq c(P)$,} \\
            0 & \text{if $J \not \subseteq c(P)$.}
        \end{cases}
    \]
    Thus, $\sco^{\alpha, \beta}_{\emptyset} (\cdot) = \sco^{\alpha, \beta}(\cdot)$.
    We extend this definition to a path collection $\pp$ in a linear fashion:
    $
    \sco^{\alpha, \beta}_J (\pp) := \sum_{P \in \pp}  \sco^{\alpha, \beta}_J (P).
    $
    Note that $\sco^{\alpha, \beta}_J (\pp) = \sco^{\alpha, \beta}_J (\pp[J])$.    
\end{definition}

\begin{definition}[Park]\label{def:park}
    A collection of paths $\pp$ is a \emph{park} with respect to some score function $\sco(\cdot)$, if for all $J \subseteq C$, it holds that $\sco_J (\pp) \leq 1$.%
    \footnote{Since $\sco_{c(P)}(P)=\beta$, we require $\beta \leq 1$, otherwise all parks w.r.t. the score function are empty.}
\end{definition}

The following lemma is the most crucial property of parks with regards to color fault-tolerance. It says that whenever \emph{a park} has enough (score-wise) paths, it has fault-tolerant properties:

\begin{lemma}[Fault-Tolerance of Parks]\label{lemma:union_bound_scores}
    Let $\pp$ be a park w.r.t.\ an $(\alpha,\beta)$-score function $\sco(\cdot)$.
    Let $F$ be a set of at most $f$ (faulty) colors. 
    If there exists $J \subseteq C$ such that $J\cap F=\emptyset$ and $\sco_J(\pp) > 1/\alpha$,
    then there is a path $P \in \pp[J]$ with $c(P) \cap F = \emptyset$ (i.e, $P$ is non-faulty).
\end{lemma}

\begin{proof}
Note that $\bigcup_{c\in F} \pp[J\cup \{c\}]$ consists of all \emph{faulty} paths in $\pp[J]$, and
\begin{align*}
    \sco_J \left( \bigcup_{c\in F} \pp[J\cup \{c\}] \right)
    &\leq \sum_{c\in F} \sco_J(\pp[J\cup \{c\}]) && \text{by union bound,} \\
    &= \sum_{c\in F} \frac{1}{\alpha f}  \sco_{J \cup \{c\}}(\pp[J\cup \{c\}]) && \text{by def.\ of $\sco$, since $c \notin J$,} \\
    &\leq \frac{1}{\alpha} && \text{as $\pp$ is a park, and $|F| \leq f$.} 
\end{align*}
Since $\sco_J(\pp) > 1/\alpha$, $\pp[J]$ must contain a non-faulty path.
\end{proof}

In fact, in our algorithm, we will only use scores with $\alpha \geq 2$. Thus, we can always apply the above lemma when the induced score of the $J$-link is more than $1/2$.
We call such a link \emph{full}:

\begin{definition}[$J$-Full Park]\label{def:full-park}
    Let $\pp$ be a park with respect to $\sco(\cdot)$, and let $J \subseteq C$ be a subset.
    We say $\pp$ is \emph{$J$-full} if $\sco_J (\pp) > 1/2$.
\end{definition}

\Cref{sec:fault_tolerant_game} further demonstrates the relation between parks and fault-tolerance, by describing an ``online fault-tolerance'' game. 
This game illustrates the $f^{k-1}$ dependence in Theorem~\ref{thm:key_results}, and captures
the essence of our algorithm's strategy for safely discarding edges.

Finally, recall that to enforce topological structure, beside having ``global'' park conditions (w.r.t. one score function), we also require ``local'' park conditions (w.r.t. a second, different score function).
This is captured by a notion we call a ``touristic park''.

\begin{definition}[Touristic Park]
    A path collection $\mathcal{P}_v$ stemming from a vertex $v \in V$ is a \emph{touristic park} with respect to global score $\gsco(\cdot)$ and local score $\lsco(\cdot)$, or a $(\gsco, \lsco)$-touristic park for short, if both following conditions hold:
    \begin{enumerate}
        \item \emph{(``Global'')} $\pp_v$ is a park with respect to $\gsco(\cdot)$.
        \item \emph{(``Local'')} For every $s \in V$, $\pp_{v,s}$ is a park with respect to $\lsco(\cdot)$.
    \end{enumerate}
\end{definition}

 \subsection{Connections to Sunflowers}\label{sec:parks_sunflowers}
We now draw some connections between color fault-tolerance, sunflowers, and parks/spread-systems.
A \emph{sunflower} is a collection of sets whose pairwise intersections are all identical.
If we think about the elements of the sets as colors, then sunflowers of $f+1$ sets can be seen as generalizing the basic fault-tolerant structure of $f+1$ disjoint sets of colors.
Essentially, if we may suffer from $f$ color faults, but have a guarantee that the kernel (the pairwise intersection) of an $(f+1)$-sunflower is non-faulty, then clearly there must be a surviving set in the sunflower.

The famous Sunflower Conjecture~\cite{ErdosRado61_sunflower} in extremal combinatorics concerns upper bounding the function $g(f,k)$, defined as the maximum possible size of a collection of sets of size $k$ that does not contain an $f$-sunflower.
Recently, a breakthrough result of \cite{LovettSpreadSetSystems20}
and follow-up works \cite{Rao20coding,Rao22,BellCW21} provided improved bounds on $g(f,k)$.
All of them use a simple inductive argument, that reduces the question to a problem regarding spread-systems.
In our parks and scores terminology, the latter problem can be stated as follows: Minimize $\alpha = \alpha(k,f)$ such that any park (of $k$-sets) w.r.t.\ $\sco^{\alpha, 1} (\cdot)$, which is full in some $J$-link, must contain an $f$-sunflower (whose kernel is $J$); such an $\alpha$ yields the upper bound $g(f,k) \leq (\alpha f)^k$.
It is easy to show that setting $\alpha = \Theta(k)$ suffices for the property.
This yields $g(f,k) = O((kf)^k)$, which is (asymptotically) the original bound given by~\cite{ErdosRado61_sunflower}.
The current record is $\alpha = \Theta(\log k)$, given by~\cite{BellCW21}.
Proving $\alpha = \Theta(1)$ will prove the Sunflower Conjecture, asserting that $g(f,k) \leq O(f)^k$.

%% file: 2k-1-alg.tex
\section{The $f$-ECFT $(2k-1)$-Spanner Algorithm}

\subsection{Description of the Algorithm}\label{sect:alg-outline}
The general outline is the same as in our description of the Baswana-Sen algorithm in \Cref{sect:baswana-sen}, working in levels $i=0,1,\dots,k-1$, with $E_i$ and $H_i$ as defined there.
A key aspect of our algorithm is the maintenance of touristic parks that accompany undecided edges.
The idea is to ensure that every edge in $E_i$ comes with two suitable touristic parks of length-$i$ paths, stemming from its endpoints, and ending at the $i$-centers $S_i$.
These would be used to provide non-faulty short-stretch paths for discarded (``safe'') edges.
For each level $i$, we will have global and local score functions $\gsco^i(\cdot)$ and $\lsco^i(\cdot)$, whose parameters are all $\Theta_k (1)$, namely, they depend only on $i,k$, and can be thought of as carefully chosen constants (when $k$ is constant).
The exact parameters are specified later, in \Cref{sect:params}.
We maintain the following invariants, analogous to~\ref{inv:BS1} and \ref{inv:BS2}:
\begin{enumerate}[label={({\bfseries I\arabic*})}]
    \item If $e = \{u,v\} \in E-E_i$, then for every set $F$ of at most $f$ colors not damaging $e$ (i.e., such that $c(e) \notin F$), it holds that $\dist_{H_i - F} (u,v) \leq (2i-1) w(e)$.\label{inv:I1}
    
    \item If $e = \{v,u\} \in E_i$, then there is a $(\gsco^i, \lsco^i)$-touristic park $\pp_u$ which consists of paths in $H_i$, each stemming from $u$, ending at some $i$-center from $S_i$, and having hop-length $i$, where each of the $i$ edges on the path weighs at most $w(e)$.
    Further, there exists $I \subseteq \{c(e)\}$ (i.e., $I = \emptyset$ or $I = \{c(e)\}$) such that $\pp_u$ is $I$-full (with respect to the score $\gsco^i(\cdot)$).%
    \footnote{Note that both $\pp_u$ and $I$ depend on $e$ and on $i$, but we exclude these from the notation to avoid clutter. Additionally, note that this invariant holds for both endpoints of $e$ (namely, also with $v$ instead of $u$).}
    \label{inv:I2}
    
\end{enumerate}
At initialization, $E_i = \emptyset$, so Invariant~\ref{inv:I1} holds vacously.
For Invariant~\ref{inv:I2}, we take $\pp_u$ to contain the unique $0$-length which starts and ends in the vertex $u$;
the $\beta$-parameter of $\gsco^0$ and $\lsco^0$ is $1$ (as discussed in \Cref{sect:params}), so this is indeed a touristic park, which is (globally) full for $I=\emptyset$.

To execute level $i$, each vertex $v \in V$ goes over its incident $E_i$-edges, denoted $E_i (v)$, in non-decreasing order of weight, and gradually partitions them into three subsets: $E_i (v, \keep)$, $E_i (v, \safe)$ and $E_i (v, \postpone)$.
All $\keep$-edges are added to our spanner, so
$
H_{i+1} = H_i \cup \bigcup_{v\in V} E_i (v, \keep).
$
The input for the next, $i+1$ level is
$
E_{i+1} = E_i - \bigcup_{v\in V} \left(E_i(v,\keep) \cup E_i (v,\safe)\right).
$
Namely, $E_{i+1}$ consists of all edges $e = \{u,v\} \in E_i$ such that $e \in E_i (v, \postpone) \cap E_i (u, \postpone)$.

Henceforth, we focus on the procedure executed in level $i$ for a single vertex $v \in V$.

\paragraph{Key Idea: Constructing the Park $\prepp_v$.}
A crucial component in the procedure for $v$ is gradually constructing a  touristic park $\prepp_v$, which consists of paths stemming from $v$ and ending at the $i$-centers $S_i$.
All paths in $\prepp_v$ will be of the form $e \circ P$, where $e = \{v,u\} \in E_i (v,\keep)$, and $P$ is a path from $\pp_u$, the park guaranteed to exist for $e$ (with endpoint $u$) by Invariant~\ref{inv:I2}.
Hence, each path we insert to $\prepp_v$ has exactly $i+1$ edges, and is supported on $H_i \cup E_i(v,\keep)$ (hence, also on the output $H$).

At the beginning of the procedure, $\prepp_v$ is initialized to be empty, and we make sure that it remains a $(\pregsco^i, \prelsco^i)$-touristic park throughout the execution.
The ``$(\alpha,\beta)$ parameters'' of the local score function $\prelsco^i(\cdot)$ are ``constants'' $(\Theta_k (1), \Theta_k (1))$.
However, for the global score function $\pregsco^i(\cdot)$, they are $(\Theta_k (1), {\Theta}_k (p / \log n) )$. 
Intuitively, the $p/\log n$ in the ``$\beta$-parameter'' causes the number of paths that are required to make a (global) link in $\prepp_v$ full, to be bigger by a factor of roughly $1/p \cdot \log n$, compared to having constant (or, more precisely, $\Theta_k (1)$) park parameters. 
(Similarly to the gap we had between the local and global ``rules'' in the description of \cite{BaswanaS07} in \Cref{sect:baswana-sen}.)

The purposes of $\prepp_v$ are threefold.
First, $\prepp_v$ is used as a mechanism to bound the number of $\keep$-edges: Each $\keep$ decision will add paths to $\prepp_v$, causing its global score to increase; but as $\prepp_v$ is a park, its score cannot exceed $1$. 
Second, the paths in $\prepp_v$ will be used to analyze $\safe$ decisions.
Specifically, the local park conditions will help to ensure the existence of non-faulty paths that ensure a $\safe$ edge can indeed be discarded.
Third, to support $\postpone$ decisions, we need to provide parks for the postponed edges so that Invariant~\ref{inv:I2} will hold for the next level $i+1$; these will be constructed from paths in $\prepp_v$ whose ending $i$-center was also sampled as an $(i+1)$-center.

\paragraph{Deciding on Edge $e \in E_i (v)$.}
We now describe how we decide if $e = \{v,u\} \in E_i(v)$ is a $\keep$, $\safe$ or $\postpone$ edge.
Recall that each $e$ has an associated park $\pp_u$ provided by Invariant~\ref{inv:I2}.
Intuitively, in order to decide on $e$, we let each path in $\pp_u$ \emph{vote} if we should declare $\keep$, $\safe$ or $\postpone$ on $e$.
The vote of $P \in \pp_u$ is determined by checking if adding  $e \circ P$ into $\prepp_v$ might get us too close to violating the touristic-park properties of $\prepp_v$.
If it would get us too close to violating the \emph{local} park property, then $P$ votes $\safe$.
If it would get us too close to violating the \emph{global} park property, then $P$ votes $\postpone$.
Otherwise, if we don't get close to a violation, $P$ votes $\keep$.
See~\Cref{fig:view-for-e} for an illustration.

\begin{figure}
    \centering
    \includegraphics[width=\textwidth]{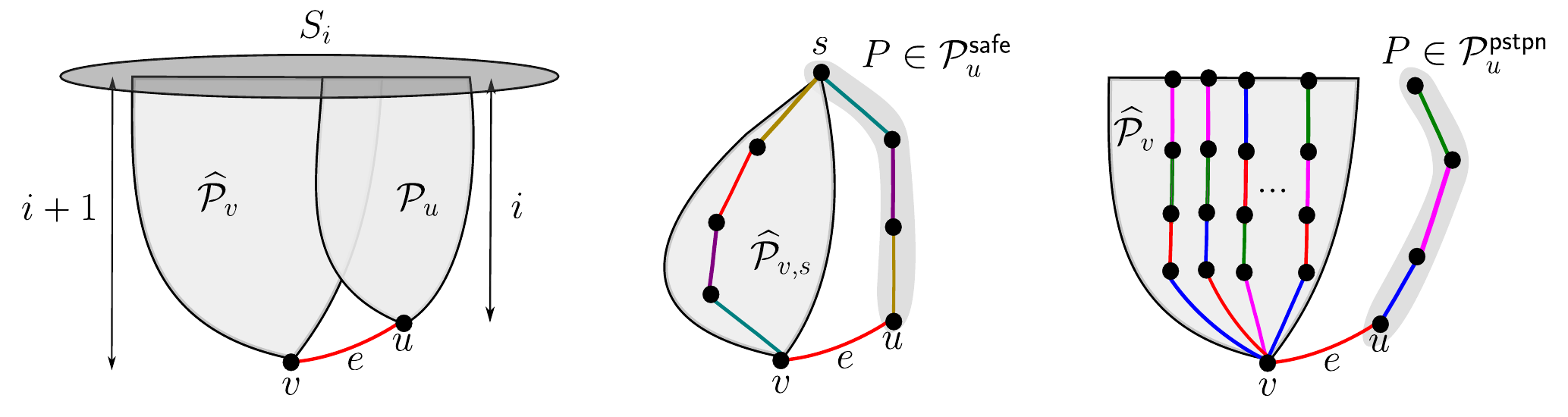}
    \caption{
    Deciding on an edge $e = \{v,u\}\in E_i(v)$.
    The left figure illustrates that $v$ maintains a touristic park $\prepp_v$ (w.r.t. $\prelsco^i, \pregsco^i$) of $(i+1)$-paths ending at $S_i$, and $e$ comes with a touristic park $\pp_u$ (w.r.t. $\lsco^i, \gsco^i$), as guaranteed by Invariant \ref{inv:I2}, of $i$-paths ending at $S_i$.
    The middle figure illustrates a special case of a $\safe$ vote of a path $P$ ending at $s\in S_i$, since there are $\Theta_k(1)$ paths (depicted as a single path) in $\prepp_{v,s}$ with the same color set as $P$, and thus
    $\prepp_{v,s}$ is $J$-full for $J=c(P)$.
    The figure on the right illustrates a $\postpone$ vote of a path $P$, caused since $\prepp_v$ contains $\Theta_k(\tfrac{\log n}{p})$ paths having the same color set as $P$, and thus $\prepp_v$
    is $J$-full for $J=c(P)$.
    }
    \label{fig:view-for-e}
\end{figure}

Formally, we go over the paths $P \in \pp_u$ and partition them into three sets, $\pp_u^{\safe}$, $\pp_u^{\postpone}$, $\pp_u^{\keep}$, as follows.
Let $s \in S_i$ be the $i$-center in which $P$ ends.
\begin{itemize}[beginpenalty=10000]
    \item \emph{(``Vote $\safe$'')} If there is $J \subseteq c(e \circ P)$ s.t.\ $\prepp_{v,s}$ is $J$-full (w.r.t.\ $\prelsco^i$), then $P \in \pp_u^{\safe}$.
    \item \emph{(``Vote $\postpone$'')} Else, if there is $J \subseteq c(e \circ P)$ s.t.\ $\prepp_v$ is $J$-full (w.r.t.\ $\pregsco^i$), then $P \in \pp_u^{\postpone}$.
    \item \emph{(``Vote $\keep$'')} Else, $P \in \pp_u^{\keep}$.
\end{itemize}

We then decide on the edge $e$ by aggregating the votes of the paths in $\pp_u$ according to their $\gsco^i_I(\cdot)$ score, where $I \subseteq \{c(e)\}$ is such that $\pp_u$ is $I$-full, which is guaranteed by Invariant~\ref{inv:I2}.
It decides to put $e$ in $E_i (v, \decision)$ where $\decision \in \{\keep, \safe, \postpone\}$ is such that $\gsco^i_I (\pp_u^{\decision}) \geq 1/8$.
Because $\pp_u$ is $I$-full, we have that
$
1/2 < \gsco^i_I (\pp_u) = \gsco^i_I (\pp_u^{\keep}) + \gsco^i_I (\pp_u^{\safe}) + \gsco^i_I (\pp_u^{\postpone}),
$
so there must exist such $\decision$.

\paragraph{Processing After Decision.}
Next, we describe the processing executed after the decision on $e$, before moving on to the next edge in $E_i (v)$.
We also state the key corresponding lemmas used for the analysis; their proofs appear in \Cref{sect:analysis}.

\begin{enumerate}
    \item \emph{Processing a $\safe$ edge:}
    If $\decision = \safe$ then, in fact, no further algorithmic processing is required.
    Our analysis shows that Invariant~\ref{inv:I1} will indeed hold when discarding $e$.
    \begin{restatable}[Safe Lemma]{lemma}{safety}\label{lem:safety}
        Suppose $e = \{v,u\} \in E_i (v)$ was decided a $\safe$ edge.
        Consider $E_i (v,\keep)$ and $\prepp_v$ in their states before (and also after) deciding on $e$.
        Denote $H'_i = H_i \cup E_i (v,\keep)$.
        Then, for every set $F$ of at most $f$ colors not damaging $e$, it holds that
        $
        \dist_{H'_i - F}(v,u) \leq (2i+1) w(e).
        $
    \end{restatable}

    \item \emph{Processing a $\keep$ edge:}
    If $\decision = \keep$, we update $\prepp_v \gets \prepp_v \cup (e \circ \pp_u^{\keep})$.
    Our analysis shows that 
    adding the path of $e \circ \pp_u^{\keep}$ to $\prepp_v$ causes an $\tilde{\Omega}_k (p/f)$ increase in the global score of $\prepp_v$, which is crucial for bounding the number of $\keep$ edges.
    \begin{restatable}[Keep Lemma]{lemma}{keepscore}\label{lem:keep_score}
        Suppose $e=\{v,u\} \in E_i (v)$ was decided a $\keep$ edge. Then
        \[
        \pregsco^i (e \circ \pp_u^{\keep}) = \Omega_k \Big( \frac{p}{f \log n} \Big).
        \]
    \end{restatable}
    
    Additionally, we prove that the touristic-park properties of $\prepp_v$ are maintained; the addition of a single $\keep$-voting path cannot violate them, but we need to show that adding \emph{all of them simultaneously} does not cause us to ``overshoot'' and violate the park score limits.
    \begin{restatable}[No Overshooting Lemma]{lemma}{noovershooting}\label{lem:no_overshooting}
        Suppose that before processing $e = \{v,u\} \in E_i (v)$, $\prepp_v$ is a $(\pregsco^i, \prelsco^i)$-touristic park.
        Assume $e$ was decided a $\keep$ edge.
        Then $\prepp_v \cup (e \circ \pp_u^{\keep})$ is also $(\pregsco^i, \prelsco^i)$-touristic.
    \end{restatable}

    \item \emph{Processing a $\postpone$ edge:}
    If $\decision = \postpone$, meaning we postpone $e$ to the next level, we must also provide it with a corresponding park stemming at $v$ that will satisfy Invariant~\ref{inv:I2} for $e$ in level $i+1$, with respect to the endpoint $v$.%
    \footnote{
    The other endpoint $u$ of $e$ is handled by the independent procedure that we execute for $u$, where $e$ is considered as an edge in $E_i (u)$.
    }
    To this end, we show the following:
    \begin{restatable}[Postpone Lemma]{lemma}{lemfullprakinnextiter}\label{lemma:full_park_in_next_iteration}
            Suppose $e = \{v,u\} \in E_i (v)$ was decided a $\postpone$ edge.
            Consider $\prepp_v$ in its state right before deciding on $e$.
            Then, with high probability, there exists a sub-park  $\pp_v \subseteq \prepp_v$ which satisfies Invariant~\ref{inv:I2} for the $v$-endpoint of $e$ in level $i+1$.
    \end{restatable}
    Notice that we slightly abuse notation, as different postponed edges incident to $v$ may have different subparks (i.e., $\pp_v$ in the lemma above actually depends on $e$).
    Although this lemma is stated existentially, it must also be implemented \emph{algorithmically}, so that we can use the new park $\pp_v$ when executing the next, $i+1$ level.
    The Postpone Lemma is the most technically involved part of the analysis.
    Its proof and algorithmic implementation are discussed in \Cref{sect:postpone,sect:sampling}.
    
\end{enumerate}

\paragraph{The Last Level.}
In the last, $k-1$ level, we of course cannot postpone our treatment of any edge $e \in E_{k-1}$.
So, a priori, it seems like we might need to adapt the execution of this level compared to the previous ones.
However, similarly to our description of the Baswana-Sen algorithm in \Cref{sect:baswana-sen}, it turns out that we can execute the $k-1$ level exactly as the previous levels, and \emph{automatically} ensure that no postpones occur, just by an appropriate choice of park parameters (specified in \Cref{sect:params}).
That is, we show:

\begin{restatable}[Last Level Lemma]{lemma}{lastlevel}\label{lem:last_level}    
    With high probability, in the $k-1$ level, no edge is postponed.
    That is, for every vertex $v \in V$ and every $e \in E_{k-1}(v)$, the decision of $v$ on $e$ is either $\keep$ or $\safe$.
\end{restatable}

We now wrap up the correctness arguments of our algorithm:

\begin{itemize}
    \item
    \textbf{Stretch:} 
    By \Cref{lem:safety}, Invariant~\ref{inv:I1} is maintained before and after every level, and particularly after the last level, i.e., for $i=k$. 
    By \Cref{lem:last_level}, it holds that $E_k = \emptyset$.
    Recall also that $H = H_k$.
    Hence, Invariant~\ref{inv:I1} guarantees that for every $e = \{u,v\} \in E$, and every set $F$ of at most $f$ colors not damaging $e$, it holds that $\dist_{H-F} (u,v) \leq (2k-1) w(e)$, which immediately implies that $H$ is an $f$-ECFT $(2k-1)$-spanner.

    \item
    \textbf{Size:}
    Consider some level $0 \leq i \leq k-1$ and vertex $v \in V$.
    After all edges in $E_i (v)$ are processed, $\prepp_v$ is the \emph{disjoint} union of $e \circ \pp_u^{\keep}$ over all edges $e = \{v,u\} \in E_i (v, \keep)$.
    Hence, $\pregsco^i(\prepp_v) \geq |E_i (v, \keep)| \cdot \Omega_k (p (f \log n)^{-1})$ by \Cref{lem:keep_score}.
    On the other hand, $\prepp_v$ is a park w.r.t.\ $\pregsco^i$ (by \Cref{lem:no_overshooting}), thus $\pregsco^i(\prepp_v) \leq 1$.
    Hence, $|E_i (v, \keep)| = O_k (p^{-1} f \log n)$.
    As the edges of the output $H$ are exactly $\bigcup_i \bigcup_v E_i (v,\keep)$, we obtain that $H$ contains $O_k (np^{-1}f \log n) = O_k (fn^{1+1/k} \log n)$ edges.
\end{itemize}

\subsection{Choice of Parameters}\label{sect:params}

Let $D$ be a sufficiently large constant. For concreteness, $D=16$ suffices for us.
We define:
\begin{equation}\label{eq:param-def}
    \begin{matrix}
        \alpha_i       & := & D^{10k(4k-2i)}   & & & \beta_i       & := & D^{-2i} \\
        \hat{\alpha}_i & := & D^{10k(4k-2i-1)} & & & \hat{\beta}_i & := & D^{-2i-5}
    \end{matrix}
\end{equation}
and
\begin{equation}\label{eq:rho-def}
    \rho := p \cdot \frac{1}{\Theta( k(\log n + k^2) )} .
\end{equation}
Finally, we define the parameters of the used score functions as follows:
\begin{equation}\label{eq:scores-def}
    \begin{matrix}
       \gsco^i (\cdot) & := & \sco^{\alpha_i, \beta_i}(\cdot)   & & & \lsco^i (\cdot)       & := & \sco^{2,\beta_i}(\cdot) \\
    \pregsco^i (\cdot) & := & \sco^{\hat{\alpha}_i, \hat{\beta}_i \rho}(\cdot) & & & \prelsco^i (\cdot) & := & \sco^{2,\beta_i / D}(\cdot)
    \end{matrix}
\end{equation}

The most important issue in this choice of parameters is in the difference between $\gsco$ and $\pregsco$ in the $\beta$-parameter --- the one in $\pregsco$ is scaled by $\rho$.
This is to ensure that if we take a park $\prepp$ ending at $S$ that is full w.r.t. $\pregsco^i$, and sample every $s\in S$ independently w.p. $p$, then we can compute a park $\pp$ that is full w.r.t. $\gsco^{i+1}$ and ends at sampled vertices.
A secondary but painful issue, is
in the difference between the global functions and the local functions in the $\alpha$ parameter, which is fixed at $2$ for the local functions, but is $2^{O(k^2)}$ for the global functions.
As was in the warm-up in \Cref{sec:warm-up}, $\alpha=2$ suffices for the safety argument.
The $2^{O(k^2)}$ factor in the global functions is due to a technical reason in \Cref{lemma:full_park_in_next_iteration}, which causes us to lose a $2^{O(k)}$ factor in every level. 
Other choices, like the differences between $\alpha$ and $\hat{\alpha}$, and between $\beta$ and $\hat{\beta}$ are relatively minor, and are mostly to satisfy \Cref{lem:no_overshooting}.

%% file: 2k-1-analysis.tex
\section{Analysis}\label{sect:analysis}

\subsection{``Safe'' Edges}
Before proving the Safe Lemma (\Cref{lem:safety}), we first observe the following property of the edge weights in $\prepp_v$:

\begin{claim}\label{lem:parks_are_mon_dec}
    Let $e = \{v,u\} \in E_i (v)$.
    Consider $\prepp_v$ in its state right \emph{after} processing $e$.
    Then all paths in $\prepp_v$ consist of edges with weight at most $w(e)$.
\end{claim}
\begin{proof}
    Invariant~\ref{inv:I2} (for level $i$) guarantees that each edge in a path of $e \circ \pp_u$ weighs at most $w(e)$. Thus, all paths that are \emph{added} to $\prepp_v$ when processing $e$ have the desired property.
    The claim follows by induction, as we 
    process the edges in $E_i(v)$ in \emph{non-decreasing} order of weight.
\end{proof}

We are now ready to prove \Cref{lem:safety}:
\safety*

\begin{proof}
Fix a set $F$ as in the lemma.
Then $F \cap I = \emptyset$, as $c(e) \notin F$ and $I \subseteq \{c(e)\}$.
As $e$ was decided $\safe$, it holds that $\gsco^{i}_{I}(\pp_u^{\safe})\geq 1/8 > 1/\alpha_i$ (by choice of $\alpha_i$).
Thus, by \Cref{lemma:union_bound_scores} (the fault-tolerant property of parks), there exists $P_1 \in \pp_u^{\safe}$ such that $c(P_1) \cap F = \emptyset$.
Recall that, by Invariant~\ref{inv:I2}, $P_1$ is supported on $H_i$ and has $i$ edges.
Let $s \in S_i$ be the $i$-center in which $P_1$ ends.
As $P_1$ voted $\safe$, there exists $J \subseteq c(e \circ P_1)$ such that $\prepp_{v,s}$ is $J$-full, meaning that $\prelsco_J (\prepp_{v,s}) > 1/2$.
We have $J \cap F \subseteq (\{c(e)\} \cup c(P)) \cap F = \emptyset$.
Hence, we can apply \Cref{lemma:union_bound_scores} on $\prepp_{v,s}$, $J$ and $F$,%
\footnote{Recall that the ``$\alpha$-parameter'' of $\prelsco^i$ is $2$.}
and find a path $P_2 \in \prepp_{v,s}$ such that $c(P_2) \cap F = \emptyset$.
Joining $P_1$ and $P_2$ at $s$ gives a path between $u$ and $v$ in $H'_i - F$ with $2i+1$ edges, all of which have weight at most $w(e)$ by \Cref{lem:parks_are_mon_dec}.
\end{proof}

\subsection{``Keep'' Edges}

Before proving the lemmas related to $\keep$ decisions, we first need a simple technical claim regarding concatenating edges to path collections.
The proof is found in \Cref{sect:parks-toolbox}.
\begin{restatable}[Concatenating an Edge]{claim}{concatedge}\label{clm:concat-edge}
    Let $\pp$ be a path collection that is stemming from an endpoint of an edge $e$.
    Let $\sco(\cdot)$ be an $(\alpha,\beta)$-score function, and let $J \subseteq C$.
    Then:
    \begin{enumerate}[label=(\arabic*)]
        \item $\sco_J(\pp) \leq \sco_{J\cup \{c(e)\}}(e \circ \pp)$.

        \item $\sco_J(e \circ \pp)\leq \sco_J(\pp) + \sco_{J-\{c(e)\}}(\pp)$. 
    \end{enumerate}
\end{restatable}

We now recall and prove the two lemmas stated for $\keep$ edges.

\keepscore*
\begin{proof}
    We have:
    \begin{align*}
        \pregsco^i(e\circ \pp_u^{\keep})
        &= \frac{1}{\hat{\alpha}_i f} \cdot \pregsco^i_{\{c(e)\}}(e\circ \pp_u^{\keep}) && \text{as each path in $e \circ \pp_u^\keep$ contains $c(e)$,} \\
        &\geq \frac{1}{\hat{\alpha}_i f} \cdot \frac{\hat{\beta}_i \rho}{\beta_{i} }  \gsco^{i}_{\{c(e)\}}(e\circ \pp_u^{\keep}) && \text{as $\hat{\alpha}_i \leq \alpha_i$,} \\
        &\geq \frac{1}{\hat{\alpha}_i f} \cdot \frac{\hat{\beta}_i \rho}{\beta_{i} }  \gsco^{i}_{I}(\pp_u^{\keep}) && \text{by \Cref{clm:concat-edge}(1), since $I \subseteq \{c(e)\}$,} \\
        &\geq \frac{1}{\hat{\alpha}_i f} \cdot \frac{\hat{\beta}_i \rho}{\beta_{i} }  \cdot \frac{1}{8} && \text{as $e \in E_i(v,\keep)$,} \\
        &= \Omega_k \Big( \frac{p}{f \log n} \Big) && \text{By choice of parameters (\Cref{eq:param-def,eq:rho-def}).}
    \end{align*}
\end{proof}

\noovershooting*
\begin{proof}
    We first address the local condition.
    Let $s \in S_i$ and $J \subseteq C$.
    We need to show that $\prelsco^i_J (\prepp_{v,s} \cup (e \circ \pp_{u,s}^{\keep})) \leq 1$.
    We split to cases:
    \begin{itemize}
        \item If $\prepp_{v,s}$ is $J$-full:
        Then, every path in $\pp_{u,s}[J-\{c(e)\}]$ must vote $\safe$. 
        Hence, the $J$-link of $e \circ \pp_{u,s}^{\keep}$ is empty, implying that $\prelsco^i_J (\prepp_{v,s} \cup (e \circ \pp_{u,s}^{\keep})) = \prelsco^i_J (\prepp_{v,s}) \leq 1$, where the last inequality holds since $\prepp_{v,s}$ is a park w.r.t.\ $\prelsco^i$.
        
        \item If $\prepp_{v,s}$ is not $J$-full:
        Then $\prelsco^i_J(\prepp_{v,s})\leq 1/2$,
        so it suffices to prove that $\prelsco^i_J (e \circ \pp_{u,s}^{\keep}) \leq 1/2$.
        We show this:
        \begin{align*}
            \prelsco^i_J(e \circ \pp_{u,s}^{\keep})
            &= \frac{1}{D}
            \lsco^i_J(e \circ \pp_{u,s}^{\keep})  && \text{by \Cref{eq:scores-def},} \\
            &\leq \frac{1}{D} \Big( \lsco^i_J(\pp_{u,s}^{\keep})+ \lsco^i_{J - \{c(e)\}}(\pp_{u,s}^{\keep}) \Big) && \text{by \Cref{clm:concat-edge}(2),} \\
            & \leq \frac{2}{D} && \text{as $\pp_{u,s}$ is a park with $\lsco^i$,} \\
            & \leq \frac{1}{2} && \text{as $D \geq 4$.}
        \end{align*}
    \end{itemize}
    We now address the global condition.
    Fix $J \subseteq C$. We show that $\pregsco^i_J (\prepp_v \cup 
    (e \circ \pp_u^{\keep})) \leq 1$.
    \begin{itemize}
        \item If $\prepp_v$ is $J$-full: 
        Then, every path in $\pp_u [J - \{c(e)\}]$ cannot vote $\keep$ (if it doesn't vote $\safe$, it must vote $\postpone$).
        Hence, the $J$-link of $e \circ \pp_u^{\keep}$ is empty, implying that $\pregsco^i_J (\prepp_v \cup (e \circ \pp_u^{\keep})) = \pregsco^i_J (\prepp_v) \leq 1$, where the last inequality holds since $\prepp_v$ is a park w.r.t.\ $\pregsco^i$.

        \item If $\prepp_v$ is not $J$-full:
        Then $\pregsco^i_J(\prepp_v)\leq 1/2$,
        so it suffices to prove $\pregsco^i_J (e \circ \pp_u^{\keep}) \leq 1/2$. We show this:
        \begin{align*}
            &\pregsco^i_J(e \circ \pp_u^{\keep}) \\
            &\leq \frac{\hat{\beta}_i \rho}{\beta_i} \Big( \frac{\alpha_i}{\hat{\alpha}_i} \Big)^k
            \gsco^i_J(e \circ \pp_u^{\keep}) && \text{as $\alpha_i \geq \hat{\alpha}_i$, $|c(P)|<k$ for $P \in \pp_u$,} \\
            &\leq \frac{\hat{\beta}_i \rho}{\beta_i} \Big( \frac{\alpha_i}{\hat{\alpha}_i} \Big)^k \Big( \gsco^i_J(\pp_u^{\keep})+ \gsco^i_{J - \{c(e)\}}(\pp_u^{\keep}) \Big) && \text{by \Cref{clm:concat-edge}(2)} \\
            & \leq \frac{\hat{\beta}_i \rho}{\beta_i} \Big( \frac{\alpha_i}{\hat{\alpha}_i} \Big)^k \cdot 2 && \text{as $\pp_{u}$ is a park with $\gsco^i$} \\
            & = 2 D^{10k^2 - 5} \cdot \rho && \text{by choice of parameters in \Cref{eq:param-def},} \\
            & \leq \frac{1}{2} && \text{
            as $\rho \leq n^{-1/k}$ and $k = O( \sqrt[3]{\log n})$.
            }
        \end{align*}
    \end{itemize}
\end{proof}

\subsection{``Postpone'' Edges}\label{sect:postpone}

In this section, we discuss the proof and algorithmic implementation of the Postpone Lemma (\Cref{lemma:full_park_in_next_iteration}), concerned with providing $\postpone$ edges with parks that satisfy Invariant~\ref{inv:I2} for the next level.

\paragraph{First Step: Proving that $\prepp_v$ is full at postpone time.}
Our first step towards this goal is proving that whenever an edge $e = \{v,u\} \in E_i (v)$ is postponed, the park $\prepp_v$ (at that time) is either $\emptyset$-full or $\{c(e)\}$-full:
\begin{lemma}\label{lem:postpone}
    Suppose $e = \{u,v\} \in E_i (v)$ was decided as a $\postpone$ edge.
    Consider $\prepp_v$ in its state before deciding on $e$.
    Then $\prepp_v$ is $T$-full w.r.t. $\pregsco^i(\cdot)$ for some $T \subseteq \{c(e)\}$.   
\end{lemma}

We first provide some intuition for \Cref{lem:postpone}, by proving an illustrative special case.
Suppose that $\pp_u$ is $\emptyset$-full, i.e., that $I=\emptyset$, and thus $\gsco^i (\pp_u^{\postpone}) > 1/8$.
Each $\postpone$-voting path $P \in \pp_u^{\postpone}$ has done so because of some set of colors $J(P) \subseteq c(e\circ P)$ for which $\prepp_v$ is $J(P)$-full.
Let us assume that $|J(P)| = 1$ for each such $P$.
Then, we claim that there must be a $P$ such that $J(P)=\{c(e)\}$, which proves \Cref{lem:postpone} with $T=\{c(e)\}$ for this specific case.

Denote $Im(J) = \{J(P) \mid P \in \pp_u^{\postpone}\}$, and assume towards contradiction that $\{c(e)\} \notin Im(J)$.
The proof is based on showing upper and lower bounds on $|Im(J)|$, which will be contradicting due to our careful choice of parameters.
The lower bound hinges on the fact that $\pp_u$ is a park, so the combined score of paths that vote $\postpone$ because of the same color cannot be too large.
Consider any color $c \neq c(e)$.
All paths $P \in \pp_u^{\postpone}$ with $J(P) = \{c\}$ belong to $\pp_u[\{c\}]$, so we can bound their combined $\gsco^i$ by $\gsco^i(\pp_u[\{c\}]) =\tfrac{1}{\alpha_i f} \gsco^i_{\{c\}} (\pp_u) \leq \tfrac{1}{\alpha_i f}$, as $\pp_u$ is a park.
Thus,
\[
\frac{1}{8} \leq \gsco^i (\pp_u^{\postpone}) \leq \sum_{\{c\} \in Im(J)}  \gsco^i (\pp_u [c]) \leq \frac{1}{\alpha_i f} \cdot |Im(J)|.
\]
The upper bound relies on the fact that the park $\prepp_v$ is $\{c\}$-full for any color $c$ that caused postponed votes, i.e., such that $\{c\} \in Im(J)$.
Observing that each path in $\prepp_v$ can belong to $\prepp_v [\{c\}]$ only for at most $i+1 \leq k$ colors $c$ appearing on it, we have
\[
1 \geq \pregsco^i (\prepp_v)
\geq \frac{1}{k} \sum_{\{c\} \in Im(J)} \pregsco^i (\prepp_v [\{c\}]) 
= \frac{1}{k} \sum_{\{c\} \in Im(J)} \frac{1}{\hat{\alpha}_i f} \pregsco^i_{\{c\}} (\prepp_v) 
\geq  \frac{1}{k \hat{\alpha}_i f} \cdot |Im(J)| \cdot \frac{1}{2}.
\]
Combining the bounds from the two parts, we get a contradiction:
\[
\frac{\alpha_i f}{8} \leq |Im(J)| \leq  2 k \hat{\alpha}_i f
\implies  
16k \geq \frac{\alpha_i}{\hat{\alpha}_i} = D^{10k} = 16^{10k}.
\]

For the general proof, we introduce the notion of ``linkful maps'':
\begin{restatable}[Linkful Map]{definition}{deflinkful}\label{def:linkful_map_new}
    Let $\pp, \pp'$ be two parks, with respective $(\alpha,\beta)$ and $(\alpha',\beta')$ score functions $\sco(\cdot)$ and $\sco'(\cdot)$.
    Let $g : \pp \to 2^C$ be a function such that every path $P \in \pp$ is mapped to a subset of its colors $g(P) \subseteq c(P)$.
    The function $g$ is called a \emph{linkful map between $\pp$ and $\pp'$}, if $\pp'$ is $g(P)$-full for every $P \in \pp$.
\end{restatable}

Then, immediately by our definition of $J(P)$ for $P \in \pp_u^{\postpone}$, we see that $e \circ P \mapsto J(P)$ is a linkful map between $e \circ \pp_u^{\postpone}$ and $\prepp_v$.
The generalized technical argument regarding linkful maps is given in the following lemma, whose proof appears in \Cref{sect:parks-toolbox}.

\begin{restatable}[Linkful Map Lemma]{lemma}{lemlinkful}\label{lem:linkful_map_new}
    Let $\pp$, $\pp'$, $\sco(\cdot)$, $\sco'(\cdot)$, be as in \Cref{def:linkful_map_new}, and suppose that there exists a linkful map $g$ between $\pp$ and $\pp'$.
    Further, assume that $\alpha \geq \alpha'$, and that each path in $\pp'$ has at most $\ell$ colors.
    Then, for every $I \subseteq C$, there exists $T \subseteq I$ such that
    \begin{equation}\label{eq:linkful}
    \sco'_T (\pp') > \min \left\{ \frac{1}{2},~ \frac{\alpha}{2^{\ell+|I|+1} \alpha'}  \sco_I (\pp) \right\}.
    \end{equation}   
\end{restatable}

Using this machinery, we can now prove \Cref{lem:postpone} in its generality.
\begin{proof}[Proof of \Cref{lem:postpone}]
    For every $P \in \pp_u^{\postpone}$, let $J(P) \subseteq c(e \circ P)$ be the set which caused $P$ to vote $\postpone$, i.e., such that $\prepp_v$ is $J(P)$-full.
    Consequently, $e \circ P \mapsto J(P)$ is a linkful map between $e \circ \pp_u^{\postpone}$ and $\prepp_v$ with respective score functions $\gsco^i$ and $\pregsco^i$.
    By applying \Cref{lem:linkful_map_new} we obtain%
    \footnote{Recall that $\prepp_v$ consists of paths of length $i+1$, and that $\alpha_i \geq \hat{\alpha}_i$.} 
    that there exists $T \subseteq \{c(e)\}$ such that:

    \begin{align*}
        \pregsco_T^i (\prepp_{v})
        &> 
        \min \left\{ \frac{1}{2},~ \frac{\alpha_{i}}{2^{(i+1)+|\{c(e)\}|+1} \hat{\alpha}_i}  \gsco_{\{c(e)\}}^i (e \circ \pp_u^{\postpone}) \right\} \\
        &= \min \left\{ \frac{1}{2},~ \frac{\alpha_{i}}{2^{i+3} \hat{\alpha}_i}  \gsco_{\{c(e)\}}^i (e \circ \pp_u^{\postpone}) \right\}.
    \intertext{
    By \Cref{clm:concat-edge}(1), $\gsco_{\{c(e)\}}(e\circ \pp_u^{\postpone}) \geq \gsco^{i}_I(\pp_u^{\postpone}) \geq \frac{1}{8}$. Thus, we continue to bound: 
    }
        &\geq \min \left\{ \frac{1}{2},~ \frac{\alpha_{i}}{2^{i+3} \hat{\alpha}_i}  \cdot \frac{1}{8} \right\} 
        = \min \left\{ \frac{1}{2},~ \frac{\alpha_{i}}{2^{i+6} \hat{\alpha}_i} \right\}.
    \end{align*}
    By our choice of $\alpha_i,\hat{\alpha}_i$ we have that $\frac{\alpha_{i}}{2^{i+6} \hat{\alpha}_i} =\frac{D^{10k}}{2^{i+6}} \geq \frac{D^{10k}}{2^{7k}} > \frac{1}{2}$,
    which implies that $\prepp_v$ is $T$-full.
\end{proof}

\paragraph{Second Step: Sampling a full park $\pp_v \subseteq \prepp_v$ ending at $S_{i+1}$.}
The second crucial step for proving \Cref{lemma:full_park_in_next_iteration} is providing a ``park sampling'' algorithm that extracts a full park $\pp_v$ ending at $(i+1)$-centers from the park $\prepp_v$ (which is full by the previous lemma).
This is formalized in the following lemma, whose proof requires some work, and is deferred to \Cref{sect:sampling}.

\begin{restatable}{lemma}{lemsampling}\label{lem:sampling_new}
    There is a randomized algorithm that, given a $(\pregsco^i,\prelsco^i)$-touristic park $\prepp$ that is $I$-full (w.r.t. $\pregsco^i$) and is ending at $S_i$, 
    outputs a subset $\pp\subseteq \prepp$ that is a
    $(\gsco^{i+1},\lsco^{i+1})$-touristic park, is ending at $S_{i+1}$, and is $I$-full (w.r.t. $\gsco^{i+1}$), w.h.p.
    The algorithm runs in $\tilde{O}_k(|\prepp[I]|)$ time.
\end{restatable}

Combining the two steps, the Postpone Lemma  (\Cref{lemma:full_park_in_next_iteration})  easily follows:

\lemfullprakinnextiter*
\begin{proof}
    By \Cref{lem:postpone}, $\prepp_v$ is $I$-full for $I\subseteq \{c(e)\}$.
    We can thus apply \Cref{lem:sampling_new}
    to construct a $(\gsco^{i+1},\lsco^{i+1})$-touristic park $\pp_v$ that is $I$-full (w.r.t. $\gsco^{i+1}$).
    Note that $\pp_v$ is supported on $H_{i+1}$, every path in $\pp_v$ has $i+1$ edges, and all the weights are $\leq w(e)$ by \Cref{lem:parks_are_mon_dec}.
    Therefore, this $\pp_v$ satisfies Invariant~\ref{inv:I2}.
\end{proof}

\subsection{The Last Level}

\lastlevel*
\begin{proof}
    Recall that each vertex is sampled into $S_{k-1}$ independently with probability $p^{k-1}$.
    Hence, w.h.p.\ it holds that $|S_{k-1}| = O(n p^{k-1} \log n) = O(1/p \cdot \log n)$ (recall that $p = n^{-1/k}$).

    Let $e = \{v,u\} \in E_{k-1} (v)$, and consider $\prepp_v$ in its state right before processing $e$.
    Recall that $\prepp_v$ is a $(\pregsco^{k-1}, \prelsco^{k-1})$-touristic park, and is the union of $\prepp_{v,s}$ over all $s \in S_{k-1}$.
    So, for every $J \subseteq C$, it holds that
    \[
    \pregsco^{k-1}_{J}(\prepp_v)
    =\sum_{s\in S_{k-1} } \pregsco^{k-1}_{J}(\prepp_{v,s}) 
    \leq \frac{\rho}{D^4} \sum_{s\in S_{k-1} } \prelsco^{k-1}_{J}(\prepp_{v,s})
    \leq \frac{\rho}{D^4} |S_{k-1}| < \frac{1}{2},
    \]
    where the last inequality is obtained by appropriately setting the constants in the definition of $\rho$ in \Cref{eq:rho-def}.
    Thus, no path in $\pp_u$ can vote $\postpone$, so $e$ must either be decided as a $\keep$-edge or as a $\safe$-edge.
\end{proof}

\subsection{Running Time}\label{sec:running-time}

\paragraph{Sequential Running Time.}
Suppose that every $(\lsco,\gsco)$-touristic park $\pp$ has an associated data structure that supports the following operations in $\tilde{O}_k(1)$ time, for every $J\subseteq C$, $s\in S_i$:
(1) insert a path to $\pp$, (2) query $\gsco_J(\pp)$, and (3) query $\lsco_J(\pp_s)$.
Such data structure can be implemented using standard techniques, say, by maintaining a dictionary whose keys are $J \subseteq C$ and pairs of $(s,J)$ where $s\in S_i$.
Consider an edge $e=\{v,u\}\in E_i(v)$.
To decide on $e$, we check the vote of every path in $\pp_u$.
Since every path in $\pp_u$ has at most $i$ colors, it has score $\geq \beta_i(\alpha_i f)^{-i}$, so, as $\pp_u$ is a park, it contains at most $\beta_i^{-1}(\alpha_i f)^{i}=O_k(f^i)$ paths.
Thus, na\"ively, each edge decision takes $O_k(f^i)$ time, which is dominated by the last level ($i=k-1$), resulting in overall running time of $\tilde{\Theta}_k(f^{k-1}m)$.
We next describe how to improve the running time to $\tilde{O}_k(m+f^{k-1}|H|)$.
To this end, we show that in level $i$, the procedure for vertex $v$ can be executed in $\tilde{O}_k (|E_i (v)| + f^i |E_i (v,\keep)|)$ time, and summing over $i=0,1,\dots,k-1$ and $v \in V$ yields the desired running time.

It is easy to augment the park data structure so that we can sample a random path from $\pp_u$, according to the distribution that assigns a path $P$ mass proportional to $\gsco_I (P)$, in $\tilde{O}_k (1)$ time.
Now, to make a decision, there is no need to check all the paths in $\pp_u$.
Instead, we sample $O(\log n)$ paths independently according to this distribution.
W.h.p. (by the Chernoff-Hoeffding bound), the fraction of $\safe$ votes in the sample (multiplied by $\gsco^i_I(\pp_u)$) approximates $\gsco_I^i(\pp_u^{\safe})$, and
the analogous statement holds for $\postpone$ decisions.
Thus, making a decision on an edge can be done in $\tilde{O}_k(1)$ time.
If the decision is $\safe$, no further computation is needed.
When there is a $\keep$ decision, 
the algorithm collects the votes from $\pp_u$, computes $\pp_u^\keep$, and
inserts $e\circ\pp_u^{\keep}$ 
to $\prepp_v$.
This requires $\tilde{O}_k(|\pp_u|) = \tilde{O}_k (f^i)$ time.

For $\postpone$ decisions, the accounting is slightly more involved.
First, note that it makes no sense to recompute a new sampled park for every edge $e$ that gets postponed when $\prepp_v$ is $\{c(e)\}$-full; if two such edges have the same color, the park constructed for the first (lower weight) edge can also serve the second edge.
Moreover, we note that once $\prepp_v$ becomes $\emptyset$-full, all the remaining edges to be considered will be postponed, and further, all of them can be served by the same sampled park.
We argue that after 
$k\hat{\alpha}_i f = \Theta_k (f)$ edges of different colors are postponed, $\prepp_v$ must be $\emptyset$-full. 
We may assume that each of these edges $e$ was postponed because $\prepp_v$ is $\{c(e)\}$-full, as otherwise, by \Cref{lem:postpone}, we already get that $\prepp_v$ is $\emptyset$-full. 
Thus, if $C'$ is this set of $k\hat{\alpha}_i f$ colors such that $\prepp_v$ is $\{c\}$-full for all $c \in C'$,
then (using the fact that each path in $\prepp_v$ contains at most $i+1 \leq k$ colors), we obtain 
\[
\pregsco^{i}(\prepp_v) 
\geq \frac{1}{k} \sum_{c \in C'} \pregsco^i (\prepp_v [\{c\}])
= \frac{1}{k\hat{\alpha}_i f} \sum_{c\in C'} \pregsco_{\{c\}}^i(\prepp_v) 
> \frac{|C'|}{2k\hat{\alpha}_i f}\geq \frac{1}{2}.
\]
All in all, the above discussion implies that the algorithm of \Cref{lem:sampling_new} needs to be executed $O_k (f)$ times with $|I| = 1$
and only once with $I = \emptyset$.
Note that $\prepp_v \subseteq \bigcup_e (e \circ \pp_u)$, where the union is over all $e \in E_i (v,\keep)$. Hence, $|\prepp_v [I]|$ is $O_k (|E_i (v,\keep)| \cdot f^{i-1})$ when $|I|=1$, and $O_k (|E_i (v,\keep)| \cdot f^i)$ when $|I|=0$.
Thus, the overall running time to process $\postpone$ decisions is $\tilde{O}_k (f^i |E_i (v,\keep)|$).

\paragraph{Distributed Implementation.}
In both the LOCAL and CONGEST models of distributed computation, the required communication in level $i$ is merely to exchange the parks guaranteed by Invariant~\ref{inv:I2}, each containing $O_k (f^{k-1})$ paths of length $O(k)$, between neighboring vertices along $E_i$-edges; all other computations are done locally in each vertex.
Thus, the algorithm requires $O(k)$ rounds in LOCAL, and $O_k (f^{k-1})$ rounds in CONGEST.

%% file: VCFT-FINAL-LEVEL.tex
\section{The VCFT Setting}
The VCFT algorithm is similar to the ECFT algorithm with some minor adaptations, except for the last level ($i=k-1$) which requires a major revision of the algorithm, described in the following \Cref{sec:VCFT_last_level}.
The high order bits of the rest of the modifications are as follows:

\begin{itemize}
    \item \textbf{Sampling probability:}
    The probability of sampling a center is set to be $p\coloneqq (n/f)^{-1/k}$.
    Thus, every vertex has a ``budget'' to store $\tilde{O}_k(f/p)$ edges, as was for the ECFT setting.

    \item \textbf{Invariant \ref{inv:I2}:} The only needed change in  \ref{inv:I2} is the definition of $I$. We require that $\{c(u)\}\subseteq I \subseteq \{c(v),c(u)\}$.
    The rest is the same, e.g., $\pp_u$ is still a $(\gsco^i,\lsco^i)$-touristic park.

    \item \textbf{Park sizes:} Observe that the touristic park $\pp_u$ guaranteed by Invariant \ref{inv:I2}
    \emph{equals} $\pp_u[c(u)]$, and for every $s\in S_i$, the local park $\pp_{u,s}$ \emph{equals} $\pp_{u,s}[\{c(u),c(s)\}]$.
    Thus, we only check the global parks conditions for $J\subseteq C$ that contain $c(u)$, since if $c(u)\notin J$, then $\gsco_J(\pp_u) = (\alpha f)^{-1}\gsco_{J\cup \{c(u)\}}(\pp_u)\leq (\alpha f)^{-1}$ and $J$ is never full.
    Similarly, we only check the local park conditions of $\pp_{u,s}$ for $\{c(u),c(s)\}\subseteq J\subseteq C$.
    Hence, global parks have $i$ ``extra colors'' while local parks have $i-1$ ``extra colors''.
    This is in contrast to the ECFT setting, where both the global and local parks are w.r.t., $i$ colors.
    Due to this fact, the maximum size (number of paths) of the local parks scales differently with $f$, as shown in \Cref{table:compare_parks_ECFT_vs_VCFT}.

\end{itemize}

In \Cref{sec:VCFT_modifications}, we give a detailed list of the minor adaptations required for the analysis.

\renewcommand{\arraystretch}{1.15}
\begin{table*}[!t]
\begin{center}
\begin{tabular}{|c|c|c|}
\hline
        & ECFT      & VCFT      \\
\hline
Global  & $f^{i}$ & $f^{i}$ \\
\hline
Local   & $f^{i}$ & $f^{i-1}$   \\
\hline
\end{tabular}
\caption{Comparison of the maximum number of paths in the local and global parks of the touristic park $\pp_u$ guaranteed by Invariant \ref{inv:I2} for level $i$
in the ECFT and the VCFT algorithms.
The dependence on $k$ and logarithmic factors are omitted. 
\label{table:compare_parks_ECFT_vs_VCFT}}
\end{center}
\end{table*}

\subsection{The Last Level}\label{sec:VCFT_last_level}

In the ECFT setting, we had an ``elegant shortcut'' for handling the last level, in which postpone decisions are clearly not acceptable.
We showed that the last level can be executed \emph{with exactly the same code} as previous levels (with the parameter $i$ being $i=k-1$), and that even though this code could, a priori, generate postpone decisions, with high probability this will not happen.
Unfortunately, this shortcut cannot be applied in the VCFT setting.
Thus, our approach is to \emph{change the code} for the last level, so that postponing is no longer an option.

\paragraph{Breaking Symmetry Between Endpoints of Edges.}
Recall that $E_{k-1}$ is the set of undecided edges which is the input of the last, $k-1$ level.
To implement this last level in the VCFT setting, we will crucially use a subtle property of $E_{k-1}$:
Every edge $e = \{v,u\} \in E_{k-1}$ was \emph{postponed by both $v$ and $u$} in the previous level, so it has parks that satisfy Invariant~\ref{inv:I2} (for $i=k-1$) stemming from both its endpoints.
Thus, we have the freedom to choose which of $e$'s endpoints will take care of it in the last level.
In order to bound the number of kept edges in the last level, we exploit this freedom in a rather delicate manner.

In the sequential setting, the symmetry-breaking hinges on \emph{sizes of color classes}, which can be computed in $O(n)$ time.
However, such computation cannot be executed efficiently in the distributed settings (LOCAL and CONGEST).
For these, we devise a modified symmetry-breaking scheme which can be computed efficiently (in terms of round complexity).
As the sequential symmetry-breaking scheme is more straightforward, we focus here on this setting, and defer the discussion of the adaptations for distributed settings to \Cref{sec:vcft-last-level-distributed}.

Let $V_c$ denote the set of all vertices with color $c$.
We put the endpoint $v$ in charge of $e$ only if $|V_{c(v)}| \leq |V_{c(u)}|$ (and ties are broken arbitrarily).
For a given vertex $v$, we now denote by $E_{k-1} (v)$ the set of all $E_{k-1}$ edges on which $v$ takes charge.
Thus, $\{E_{k-1} (v)\}_{v \in V}$ 
is a partition of $E_{k-1}$, such that if $e = \{v,u\} \in E_{k-1} (v)$, then $|V_{c(v)}| \leq |V_{c(u)}|$.

\paragraph{Execution of the Last Level.}
Similarly to the previous levels, we let each vertex $v$ process the edges in $E_{k-1} (v)$ in non-decreasing order of weight and partition them into $E_{k-1}(v, \keep)$ and $E_{k-1}(v,\safe)$ (this time there is no option to postpone), while gradually constructing the path collection $\prepp_v$.
However, we no longer insist that $\prepp_v$ is globally a park.
We only maintain the property that for each $s \in S_{k-1}$, $\prepp_{v,s}$ is (locally) a park with respect to $\prelsco^{k-1}$.
When processing $e = \{v,u\} \in E_{k-1} (v)$, we partition the $\pp_u$ that is guaranteed by Invariant~\ref{inv:I2} for $e$ into $\pp_u^{\safe}$ and $\pp_u^{\keep}$ by the following voting process.
For $P \in \pp_u$ ending at the center $s \in S_{k-1}$:
\begin{itemize}
    \item \emph{(``Vote $\safe$'')} If there is $J \subseteq c(e \circ P)$ s.t.\ $\prepp_{v,s}$ is $J$-full (w.r.t.\ $\prelsco^{k-1}$), then $P \in \pp_u^{\safe}$.
    \item \emph{(``Vote $\keep$'')} Else, $P \in \pp_u^{\keep}$.
\end{itemize}
Again, $v$ puts $e$ in $E_i (v, \decision)$ where $\decision \in \{\keep, \safe\}$ is such that $\gsco^{k-1}_I (\pp_u^{\decision}) \geq 1/8$.
In case $\decision = \keep$, we also update $\prepp_v \gets \prepp_v \cup (e \circ \pp_u^{\keep})$.

\paragraph{Analysis.}
It is easily verified that the local part of the No Overshooting Lemma (\Cref{lem:no_overshooting}) still goes through, so the local park property of $\prepp_v$ is maintained.
The Safety Lemma (\Cref{lem:safety}) also goes through seamlessly, ensuring that the output is indeed an $f$-VCFT $(2k-1)$-spanner of $G$.

The part of the analysis that requires significant modification is bounding the number of $\keep$ decisions; in previous levels, this argument hinged on $\prepp_v$ being a global park, which is no longer true.
Particularly, we should show that $|E_{k-1} (v, \keep)| = O_k(f/p \cdot \log n) = O_k (f^{1-1/k} n^{1/k} \log n)$, with high probability.

In fact, in the sequential setting (on which we are currently focused), we will show that $|E_{k-1} (v, \keep)| = O(f/p)$ \emph{in expectation}.
This suffices by a standard repetition argument; if by the end of the execution the size of the output $H$ exceeded twice its expected value, we can rerun the algorithm, and w.h.p.\ within $O(\log n)$ repetitions we find a spanner with the correct size bound.

We partition the edges $e \in E_{k-1} (v, \keep)$ into two \emph{types}.
Consider $e = \{v,u\} \in E_{k-1} (v, \keep)$, and let $\{c(u)\} \subseteq I \subseteq \{c(v),c(u)\}$ be the set for which $\pp_u$ is $I$-full guaranteed by Invariant~\ref{inv:I2}.
If $|I|=1$, i.e., $I = \{c(u)\}$, we say that $e$ is of Type-$1$, and if $|I|=2$, i.e., $I=\{c(v),c(u)\}$ and $c(v)\neq c(u)$, we say it is of Type-$2$.
We denote the set of Type-$j$ edges by $E_{k-1} (v,\keep,j)$.

The easier part of the analysis is bounding the number of Type-$1$ edges.

\begin{lemma}\label{lem:vcft-type-1}
    It holds that $|E_{k-1}(v,\keep,1)|$ is at most $O(f/p)$ in expectation, and at most $O(f/p \cdot \log n)$ with high probability.
\end{lemma}
\begin{proof}
    We consider how the following quantity increases during the execution:
    \[
    \Phi := \sum_{s \in S_{k-1}} \prelsco^{k-1}_{\{c(v),c(s)\}} (\prepp_{v,s}) ~.
    \]
    As each $\prepp_{v,s}$ is locally a park throughout the execution, each term in the sum is at most $1$, so we have that $\Phi \leq |S_{k-1}|$, and $|S_{k-1}|$ is $np^{k-1} = O(f/p)$ in expectation and $O(np^{k-1} \log n) = O(f/p \cdot \log n)$ with high probability.
    Whenever we keep a Type-$1$ edge $e = \{v,u\}$, the increase in $\Phi$ is
        \begin{align*}
        &\sum_{s \in S_{k-1}} \prelsco^{k-1}_{\{c(v),c(s)\}}(e\circ \pp^{\keep}_{u,s}) \\
        &\geq \sum_{s \in S_{k-1}} \prelsco^{k-1}_{\{c(s)\}}(\pp^{\keep}_{u,s}) && \text{by \Cref{clm:concat-edge}(1),} \\
        &= \sum_{s \in S_{k-1}} \prelsco^{k-1}_{\{c(u)\}}(\pp^{\keep}_{u,s}) && \text{since $c(u),c(s) \in c(P)$ for $P\in \pp^{\keep}_{u,s}$,} \\
        &= \prelsco^{k-1}_{\{c(u)\}}(\pp^{\keep}_u) && \text{by linearity of score,} \\
        &\geq \frac{1}{D} \gsco^{k-1}_{\{c(u)\}}(\pp^{\keep}_u)&& \text{by choice of parameters (\Cref{eq:param-def}),} \\
        &\geq \Omega(1) && \text{as $\gsco_{\{c(u)\}}^{k-1}(\pp^{\keep}_u) \geq 1/8$, since $e$ is kept.}
    \end{align*}
    The lemma follows.
\end{proof}

We now move to handle the harder analysis of Type-2 edges.
\begin{lemma}\label{lem:vcft-type-2}
    In expectation, $|E_{k-1}(v,\keep,2)| = O(f/p)$.
\end{lemma}
\begin{proof}
    Let $X = S_{k-1} - V_{c(v)}$ be the set of centers whose color is different than $c(v)$, and $Y = S_{k-1} \cap V_{c(v)}$ be the set of centers of the same color as $v$.
    We use $\pp^2$ to denote the restriction of a park $\pp$ to paths whose first edge is of Type-$2$.
    We will consider how the following two quantities increase during the execution:
    \begin{align*}
        \Phi_X & := \sum_{x \in X} \prelsco^{k-1}_{\{c(v), c(x)\}} (\prepp_{v,x}^2) 
        ~, \\
        \Phi_Y & := \sum_{y\in Y} \prelsco^{k-1}_{\{c(v)\}} (\prepp_{v,y}^2) ~.
    \end{align*}
    We start by analyzing the increases caused by keeping Type-$2$ edges: \begin{claim}\label{clm:potential_increase}
        With each Type-$2$ kept edge, either $\Phi_X$ increases by $\Omega (1)$, or $\Phi_Y$ increases by $\Omega (1/f)$.
    \end{claim}
    \begin{proof}
        Let $e = \{v,u\}$ be a Type-$2$ kept edge, hence
        \begin{align*}
            \frac{1}{8} &\leq \gsco^{k-1}_{\{c(v),c(u)\}}(\pp_u^\keep) \\
            & \leq D \cdot \prelsco^{k-1}_{\{c(v),c(u)\}} (\pp_u^\keep) && \text{by choice of param.\
            (\Cref{eq:param-def}),} \\
            &\leq D \cdot \prelsco^{k-1}_{\{c(v),c(u)\}} (e \circ \pp_u^\keep)  && \text{by \Cref{clm:concat-edge}(1),} \\
            &=D \Big( \sum_{x \in X}  \prelsco^{k-1}_{\{c(v),c(u)\}} (e \circ \pp_{u,x}^\keep) 
            + \sum_{y \in Y} \prelsco^{k-1}_{\{c(v),c(u)\}} (e \circ \pp_{u,y}^\keep) \Big) && \text{as $S_{k+1} = X \cup Y$, $X\cap Y=\emptyset$,}
        \end{align*}
        Thus, one of the sums in the parenthesis in the last line must be at least $1/8D = \Omega (1)$.
        
        If it is the first sum, we obtain
            \[
            \Omega (1) = \sum_{x \in X}  \prelsco^{k-1}_{\{c(v),c(u)\}} (e \circ \pp_{u,x}^\keep) = \sum_{x \in X}  \prelsco^{k-1}_{\{c(v),c(x)\}} (e \circ \pp_{u,x}^\keep)
            \]
            where the last equality holds because $c(v),c(u)$ are two different colors, and so are $c(v),c(x)$, and they all appear on all paths.
            The RHS of the above is precisely the increase in $\Phi_X$ caused by keeping $e$, so we are done.

        If it is the second sum, we obtain
            \[
            \Omega (1) = \sum_{y \in Y}  \prelsco^{k-1}_{\{c(v),c(u)\}} (e \circ \pp_{u,y}^\keep) = 2f \cdot \sum_{y \in Y} \prelsco^{k-1}_{\{c(v)\}} (e \circ \pp_{u,y}^\keep)
            \]
            where the last equality holds because $c(u) \neq c(v)$.
            The RHS of the above is $2f$ times the increase in $\Phi_Y$ caused by keeping $e$, so we are again done.
    \end{proof}

    Next, we bound $\Phi_X, \Phi_Y$ from above:
    \begin{claim}\label{clm:potential_bounds}
        It holds that:
        \begin{itemize}
            \item $\Phi_X \leq O(f/p)$ in expectation, and $\leq O(f/p \cdot \log n)$ with high probability.
            \item $\Phi_Y \leq O(1/p)$ in expectation.
        \end{itemize}
    \end{claim}
    \begin{proof}
        We start with the easier $\Phi_X$.
        Each term in the sum defining $\Phi_X$ is a score of a local park, so it is at most $1$.
        Thus, $\Phi_X \leq |X|$.
        But $|X| \leq |S_{k-1}|$,
        and $|S_{k-1}|$ is $np^{k-1} = O(f/p)$ in expectation and $O(np^{k-1} \log n) = O(f/p \cdot \log n)$ with high probability.

        We now analyze $\Phi_Y$.
        Consider any single term in the sum defining $\Phi_Y$, corresponding to some $y \in Y$, namely $\prelsco^{k-1}_{\{c(v)\}} (\prepp_{v,y}^2)$.
        At any point during the execution, the first edge of each path in $\prepp_{v,y}^2$ is from $E_{k-1}(v)$.
        Since this first edge is of Type-$2$ and by definition of $E_{k-1} (v)$, this edge
        contains a color $c \neq c(v)$ such that $|V_c| \geq |V_{c(v)}|$.
        Denote the set of such colors by 
        \[
        R = \{c \in C \mid c \neq c(v) \text{ and } |V_c| \geq |V_{c(v)}|\} ~.
        \]
        We obtain that
        \begin{align*}
            \prelsco^{k-1}_{\{c(v)\}} (\prepp_{v,y}^2)
            & \leq \sum_{c \in R}  \prelsco^{k-1}_{\{c(v)\}} (\prepp_{v,y}^2 [\{c\}]) && \text{by a union bound,} \\
            & = \frac{1}{2f}  \sum_{c \in R}  \prelsco^{k-1}_{\{c(v), c\}} (\prepp_{v,y}^2 [\{c\}]) && \text{as $c \neq c(v)$,} \\
            & \leq \frac{|R|}{2f} && \text{since $\prepp_{v,y}^2$ is a park.}
        \end{align*}
        Thus, we have that $\Phi_Y = O(|Y| \cdot |R| \cdot 1/f)$.
        As $Y$ is the set of vertices from $V_{c(v)}$ that were sampled as $(k-1)$-centers, we have $|Y| = |V_{c(v)}|\cdot p^{k-1}$ in expectation.%
        \footnote{
            This is the only point in the size analysis of $|E_{k-1} (v,\keep)|$ where we cannot multiply the expectation with $O(\log n)$ to get a ``w.h.p.\ gaurantee'' (using Chernoff), since $|V_{c(v)}|$ might be much smaller than $n$.
        }   
        Also, as the color classes partition the $n$ vertices, there cannot be more than $n/|V_{c(v)}|$ colors $c$ such that $|V_c| \geq |V_{c(v)}|$, meaning that $|R| \leq n/ |V_{c(v)}|$.
        Combining the bounds on $|Y|$ and $|R|$, we obtain that, in expectation, $\Phi_Y = O(n p^{k-1} f^{-1}) = O(1/p)$ as needed.
    \end{proof}

    The lemma follows directly from combining \Cref{clm:potential_bounds} and \Cref{clm:potential_increase}.
\end{proof}

\paragraph{Sequential Running Time.}
The sequential running time is analyzed identically to \Cref{sec:running-time}, showing that the last level can be executed in $\tilde{O}_k (m + f^{k-1}|H|)$ time.
The repeated executions of the algorithm to ensure that the output spanner has the correct size with high probability (rather than in expectation) increase the total running time only by a factor of $O(\log n)$.

\subsection{Modifications for Last Level in Distributed Settings}\label{sec:vcft-last-level-distributed}

As explained in the previous section, the main reason that distributed settings require adaptation lies in the symmetry-breaking that decides which endpoint of an edge in $E_{k-1}$ takes charge on it.
Additionally, the ``repetition trick" where we re-run the entire algorithm if the produced spanner turned out to have significantly more edges than expected is no longer applicable, since counting the total number of edges in the spanner is a non-local task.
As it turns out, our resolution of the former (symmetry-breaking) issue allows us to easily address the latter, and argue that the output spanner will be sufficiently sparse with high probability in a single execution.

\paragraph{Symmetry-Breaking in Distributed Settings.}
Instead of using sizes of color classes (which cannot be computed locally), we use the sizes of certain sets of centers of the same color, as described next.
For a vertex $v$, let $\tilde{Y}_v$ be the set of centers $s \in S_{k-1}$ such that $c(s) = c(v)$ \emph{and} $s \in P \in \pp_u$ for some edge $e=\{v,u\} \in E_{k-1}$.
Note that $v$ only needs to know $\{\pp_u \mid e=\{v,u\} \in E_{k-1} \}$ to locally compute $\tilde{Y}_v$, and this exchange of park information along edges is already accounted for in our analysis of the round complexity in both LOCAL and CONGEST (see ``Distributed Implementation" paragraph at the end of \Cref{sec:running-time}).
For $e = \{v,u\} \in E_{k-1}$,
we put the endpoint $v$ in charge of $e$ only if $|\tilde{Y}_v| \leq |\tilde{Y}_u|$ (ties broken arbitrarily).
Determining the endpoint in charge requires only $O(1)$ rounds (in CONGEST or in LOCAL) for exchanging $|\tilde{Y}_v|$ and $|\tilde{Y}_u|$ along $e=\{v,u\}$.
Again, we denote by $E_{k-1} (v)$ the set of all $E_{k-1}$ edges on which $v$ takes charge.
Now we have the following property:
\begin{observation}\label{obs:symmetry-breaking}
    If $e = \{v,u\} \in E_{k-1} (v)$, then the number of centers $s\in S_{k-1}$ whose color is $c(u)$ is at least $|\tilde{Y}_v|$.
\end{observation}

\paragraph{Execution of Last Level}
Once the responsibility for edges in $E_{k-1}$ is partitioned between their endpoints, the execution of the last level is exactly as previously described for the sequential setting (only with the modified definition of $E_{k-1}(v)$).
This entire process is executed \emph{locally} in each vertex $v$, so it requires no additional rounds of communication.

\paragraph{Analysis.}
The only part in our analysis for the sequential setting that hinges on the symmetry-breaking and requires modification is that of the second item in \Cref{clm:potential_bounds} within \Cref{lem:vcft-type-2}, namely to show that $\Phi_Y$ is at most $O(1/p)$ in expectation.
We will next show that this remains true with the modified symmetry-breaking, and furthermore, that $\Phi_Y \leq O(1/p \cdot \log n)$ with high probability. 
The latter fact immediately yields that now, in \Cref{lem:vcft-type-2} we can also say that $|E(v,\keep,2)| = O(f/p \log n)$ with high probability.
Thus, combining with \Cref{lem:vcft-type-1} (which is unaffected by the modifications), we obtain that w.h.p.\ $|E(v,\keep)| = O(f/p \log n)$, as required.

So, to conclude the analysis, it remains to show the aforementioned modifications for the proof of the second item in \Cref{clm:potential_bounds} which pertains to upper bounding $\Phi_Y$.
For this, we first observe that $\Phi_Y$ can be expressed by summing only over $y \in \tilde{Y}_v$ instead of $y \in Y$, i.e.,
\begin{align}\label{eq:modified-potentail}
    \Phi_Y 
    & \bydef \sum_{y \in Y} \prelsco^{k-1}_{\{c(v)\}} (\prepp_{v,y}^2) 
    = \sum_{y \in \tilde{Y}_v} \prelsco^{k-1}_{\{c(v)\}} (\prepp_{v,y}^2) ~.
\end{align}
This is because the parks $\prepp_{v,y}^2$ consist only of paths with the form $e \circ P$ where $e = \{v,u\} \in E_{k-1} (v)$ and $P \in \prepp_u$, hence only centers $y \in \tilde{Y}_v \subseteq Y$ can have non-empty $\prepp_{v,y}^2$.
Next, we modify the definition of $R$ to match our modified symmetry-breaking, as follows:
\[
R = \{c \in C \mid c \neq c(v) \text{ and there exists at least $|\tilde{Y}_v|$ centres from $S_{k-1}$ of color $c$} \} ~.
\]
With this modified definition, we can use \Cref{obs:symmetry-breaking} to prove, in a similar manner to \Cref{clm:potential_bounds}, that each summand in the RHS of \Cref{eq:modified-potentail} is $\leq \frac{|R|}{2f}$, implying that $\Phi_Y \leq |\tilde{Y}_v| \cdot |R| \cdot 1/(2f)$.
So, if $|\tilde{Y}_v| = 0$, then $\Phi_Y = 0$ and we are done.
Otherwise, we use the fact that $|R| \leq |S_{k-1}| / |\tilde{Y}_v|$, since there can be at most $|S_{k-1}| / |\tilde{Y}_v|$ colors that appear on at least $|\tilde{Y}_v|$ of the centers in $S_{k-1}$.
We thus obtain that $\Phi_Y = O(|S_{k-1}| \cdot 1/f)$.
As each vertex is sampled into $S_{k-1}$ independently w.p.\ $p^{k-1}$, it holds w.h.p.\ that $|S_{k-1}| = O(np^{k-1} \log n ) = O(f/p \cdot \log n)$, hence $\Phi_Y = O(1/p \cdot \log n)$ as needed.

\subsection{The Non-Last Levels}\label{sec:VCFT_modifications}
We now describe the adaptations for the analysis of all but the last levels ($i=0,1\dots,k-2$) in the VCFT setting.
Some claims hold as they are, like the 
Safety Lemma (\Cref{lem:safety}), and 
\Cref{lem:sampling_new}.
However, some claims require some minor adaptations, which are provided in the following list.

\begin{itemize}
    \item \textbf{Initialization:} At initialization, to fulfill Invariant~\ref{inv:I2} for $i=0$, we still take $\pp_u$ to contain only the $0$-length path which starts and ends at $u$; this makes $\pp_u$  $\{c(u)\}$-full w.r.t.\ $\gsco^0$.
    
    \item \textbf{Concatenating edges:} When concatenating an edge $e=\{v,u\}$ to a park $\pp_u$, only the color $c(v)$ may be added to all the paths of $\pp_u$ (as they already had $c(u)$).
    Thus, we can apply Edge Concatenation (\Cref{clm:concat-edge}) but replace $c(e)$ with $c(v)$.
    This is stated 
    formally in \Cref{claim:concat_edge_vertex_colors}.

    \item \textbf{Keep edges:} In the Keep Lemma (\Cref{lem:keep_score}), the score is now relative to $c(v)$ (instead of $\emptyset$ as in the original statement). That is:
    If $e=\{v,u\}\in E_i(v)$ was decided a $\keep$ edge, then $\pregsco^i_{\{c(v)\}}(e\circ \pp_u^{\keep}) = \Omega_k(\tfrac{p}{f\log n})$.
    The proof proceeds easily using the adapted version of Edge Concatenation (\Cref{claim:concat_edge_vertex_colors}). 

    \item \textbf{No overshooting:} The No Overshooting Lemma (\Cref{lem:no_overshooting}) holds as is, but we need to adapt the upper bound we assume on $k$, and demand $k=O(\sqrt[3]{\log (n/f)})$.

    \item \textbf{Linkful maps:} The Linkful Map Lemma (\Cref{lem:linkful_map_new}) holds as is, but with an additional factor of $2$ when applied, because a path of length $i$ may have $i+1$ colors. 

    \item \textbf{Full $\prepp_v$ at postpones:}
    Proving that $\prepp_v$ is full at postpone time, we adapt \Cref{lem:postpone} so that $\prepp_v$ is $T$-full for $\{c(v)\}\subseteq T \subseteq \{c(v),c(u)\}$.

    To this end, we note that when applying the Linkful Map Lemma (\Cref{lem:linkful_map_new}) in the proof, we get that there exists $T \subseteq \{c(v),c(u)\}$ for which there is a certain lower bound on $\pregsco^i_T(\prepp_v)$.
    We can assume without loss of generality that $c(v)\in T$, since $\pregsco^i_{T\cup \{c(v)\}}(\prepp_v)\geq \pregsco^i_T(\prepp_v)$ (we can replace $T$ with $T\cup \{c(v)\}$).
    This implies \Cref{lemma:full_park_in_next_iteration}, but with $\{c(v)\}\subseteq I\subseteq \{c(v),c(u)\}$, as required by our modified Invariant~\ref{inv:I2} in the VCFT setting.

    \item \textbf{Running time analysis:}
    Regarding the calls to the algorithm of \Cref{lemma:full_park_in_next_iteration}, in every level $i$, we execute this algorithm $O_k(f)$ times with $|I|=2$, and once with $I=\{c(v)\}$, but since a path of length $i+1$ has $i+2$ vertices, this results in the same running time.

\end{itemize}

%% file: VFT-alternative.tex
\section{An Alternative View of Parter's VFT Spanner Algorithm}\label{sect:VFT-alternative}

We now describe Parter's algorithm for VFT spanners \cite{Parter22} from our edge-centric perspective for Baswana-Sen of \Cref{sect:baswana-sen}.
The framework is exactly as in \Cref{sect:baswana-sen}, only now the sampling probability is $p = (n/f)^{-1/k}$. The following invariants are maintained for every $i \in \{0,1,\ldots, k\}$:
\begin{enumerate}[label={({\bfseries P\arabic*})}]
    \item If $e = \{v,u\} \in E - E_i$, then for every $F \subseteq V - \{u,v\}$ with $|F|\leq f$, it holds that $\dist_{H_i - F} (u,v) \leq (2i-1) w(e)$. \label{inv:P1}
    
    \item If $e = \{v,u\} \in E_i$, then $H_i$ contains a collection $\pp_u$ of $8kf$ $u$-to-$S_i$ paths of hop-length $i$, that are mutually vertex disjoint except for the starting vertex $u$, and all of their edges have weight at most $w(e)$. \label{inv:P2}
\end{enumerate}
At initialization ($i=0$), Invariant~\ref{inv:P1} holds vacuously, as $E-E_0 = \emptyset$.
For Invariant~\ref{inv:P2}, we take $\pp_u$ to contain $8kf$ copies of the path which is just the vertex $u \in S_0$ (with $0$ edges).

During level $i$,
each $v \in V$ processes its incident edges in $E_i$, denoted $E_i(v)$, in non-decreasing order of weight.
It gradually constructs a collection  $\prepp_v$ of $v$-to-$S_i$ paths, that will be all supported on $H$, according to the following rules:
(i)
the paths are mutually vertex disjoint except for their common starting vertex $v$,
and
(ii)
there are at most $O(kf/p \cdot \log n)$ paths in total.
(Interestingly, one might regard both as ``global'' rules.)
An edge $e = \{v,u\} \in E_i(v)$ is processed as follows.

\begin{itemize}
    \item \emph{(``Postpone'')} If rule (ii) is already saturated, meaning $|\prepp_v| = \Omega(kf/p \cdot \log n)$, then $v$ decides to postpone $e$.
    In this case, with high probability, there is a subset $\pp_v \subseteq \prepp_v$ of $8kf$ (mutually vertex disjoint) paths that end in $(i+1)$-centers, so we can provide $e$ with $\pp_v$ to maintain Invariant~\ref{inv:P2} for the next, $i+1$ level.

    \item \emph{(``Keep'')} Else, if there exists some $P \in \pp_u$ that is vertex disjoint from all paths in $\prepp_v$, then $v$ decides to keep $e$ in $H$, and adds $e \circ P$ into $\prepp_v$.
    Also, if $u$ appears on some path in $\prepp_v$, then $v$ keeps $e$ (without changing $\prepp_v$).
    Note that the rules guarantee that $v$ only keeps $O(k^2 f/p \cdot \log n)$ edges in total.%
    \footnote{
        Here, we assume that $G$ is a simple graph, so every vertex appearing in $\prepp_v$ can cause only one ``keep'' decision that does not add a new path to $\prepp_v$.
        This assumption is without loss of generality in the VFT setting, but not in the EFT setting; see \cite[Appendix A]{PST24spanners} for further discussion of FT spanners for simple graphs vs.\ multi-graphs.
    }

    \item \emph{(``Safe'')} Else, it holds that every $P \in \pp_u$ intersects some path in $Q \in \prepp_v$, and that $u$ does not appear in $\prepp_v$.
    In this case, $v$ decides that $e$ is safe to be discarded.    

    We now show that when this event occurs, it was indeed safe to discard $e$.
    Consider the following iterative process:
    While $\pp_u$ is not empty, find an intersecting pair $P \in \pp_u$ and $Q \in \prepp_v$, then delete from $\pp_u$ every path intersecting $Q$.
    As the paths in $\pp_u$ are vertex disjoint (except in $u$), and each path in $\prepp_v$ has $\leq 2k$ vertices (all different from $u$), we only delete $\leq 2k$ paths.
    Thus, the process runs for at least $4f$ iterations.
    Denote by $(P_1, Q_1), \dots, (P_{4f}, Q_{4f})$ the pairs of paths found in these iterations.
    Note that the deletions guarantee that $P_j \neq P_{j'}$ and $Q_j \neq Q_{j'}$ for every $j \neq j'$.
    Observe that all these paths are contained in the current $H$.
    Further, all their edges have weight at most $w(e)$ (recall that the $Q_j$s are of the form $e' \circ P'$ with $P' \in \pp_{u'}$, for some $e' = \{v,u'\} \in E_i (v)$ with $w(e') \leq w(e)$).
    For each pair $(P_j, Q_j)$, we get a corresponding $u$-to-$v$ path $W_j$ of at most $2i+1$ edges in $H$, by concatenating (prefixes of) $P_j, Q_j$ at an intersection point.

    Now, let $F \subseteq V - \{u,v\}$ be a faulty set with $|F| \leq f$. 
    As the $Q_j$'s are mutually disjoint (except for $v$), and the $P_j$'s are mutually disjoint (except for $u$), each $x \in F$ can appear in at most two of the paths $W_1, \dots, W_{4f}$.
    Hence at least one $W_j$ survives in $H-F$, which proves that $\dist_{H-F} (u,v) \leq (2i+1) w(e)$.
    Thus, Invariant~\ref{inv:P1} is maintained for the next, $i+1$ level.
    
\end{itemize}

In the last, $k-1$ level, $|S_{k-1}| = O(n p^{k-1} \log n) = O(f/p \cdot \log n)$ with high probability.
As rule (i) guarantees that each path in $\prepp_v$ ends in a unique $(k-1)$-center,
rule (ii) cannot be saturated, so there are no postponed edges and $E_k = \emptyset$.
Thus, by Invariant~\ref{inv:P1}, $H = H_k$ is indeed a $(2k-1)$-spanner of $G$.
In each level, each vertex only adds $O(k^2 f/p \cdot \log n) = O(k^2 f^{1-1/k} n^{1/k} \log n)$ edges, so $H$ has  $O(k^3 f^{1-1/k} n^{1+1/k} \log n)$ edges in total.

\paragraph{Running Time.}
As described above, the algorithm runs in sequential $O(k m \cdot k^2 f)$ time, and in $O(k \cdot k^2 f)$ CONGEST rounds.\footnote{
The first $k$ is the number of levels, and $k^2 f$ is the total number of edges in all the paths of $\pp_u$ from Invariant~\ref{inv:P2}.
}
The $k^2 f$ factor is because whenever we process $e = \{v,u\} \in E_i (v)$, we go over all $8kf$ length-$k$ paths in $\pp_u$ to determine if there is one among them who is disjoint of all paths in $\prepp_v$. (In CONGEST, this means we send all of $\pp_u$ to $v$.)
To reduce this $k^2f$ factor to $O(k \cdot \log n)$, we make the following modification:
For $O(\log n)$ independent trials, we sample\footnote{Note that a sampling operation takes $O(1)$ time, by storing $\pp_u$ as an array of pointers to its paths.} a uniformly random path $P \in \pp_u$ and check if $P$ does not intersect any path in $\prepp_v$.
If we succeed in finding such $P$, we can keep $e$ and add $e \circ P$ to $\prepp_v$.
In the complementary case, by Chernoff, w.h.p.\ at least $6kf$ of the $8kf$ paths in $\pp_u$ intersect paths in $\prepp_v$,
so we can declare that $e$ is (w.h.p) safe (the argument goes through also with $6kf$ paths).
Hence, we obtain sequential $O(k^2 m \log n)$ running time, and $O(k^2 \log n)$ rounds in CONGEST.

%% file: Fault-Tolerance-Game.tex
\section{Parks and Fault-Tolerance}\label{sec:fault_tolerant_game}

The main goal of this section is to take a step back from the details and technicalities of our algorithm, and describe why parks (or spread set-systems) naturally arise as fault-tolerant structures.
We do this via a ``toy example'' of a simple online game.

The $(f,k)$-FT game is played between two players, Alice and Bob, over a universe $C$ whose elements are considered potentially faulty (these correspond to the colors in the graph).
In each round $i = 1,2,3\dots$, Bob presents Alice with a set $P_i \subseteq C$, $|P_i| \leq k$ (corresponding to a color-set of a short path).
Alice maintains a sub-collection of the sets presented so far, in an online manner: when $P_i$ is presented, she may either include it (forever) in her sub-collection, or irrevocably discard it (with no possibility of regret).
The collection of all sets presented after round $i$ is denoted by $\pp^{\all}_i = \{P_1, P_2, \dots, P_i\}$, and the collection of sets kept by Alice is denoted by $\pp_i \subseteq \pp^{\all}_i$.
Alice's goal is to make $\pp_i$ as small as possible, while ensuring that at any time, $\pp_i$ is a \emph{FT-certificate} of $\pp^{\all}_i$:
For every set of faults $F \subseteq C$ with $|F| \leq f$, if there exists a non-faulty set $P \in \pp^{\all}_i$ (i.e., such that $P \cap F = \emptyset$), then there also exists a non-faulty set $P' \in \pp_i$.

We first show a lower bound of $\binom{f+k}{k} = \Omega_k (f^k)$ for the certificate size.

\begin{claim}[Forcing Bob]
    Bob has a strategy that guarantees $|\pp_i| = \binom{f+k}{k} = \Omega_k(f^k)$ eventually, assuming $|C| \geq f+k$.
\end{claim}
\begin{proof}
    Suppose $|C| = f+k$, otherwise ignore redundant elements.
    Bob simply sends (one by one) the sets in $\binom{C}{k}$, i.e., all subsets of $C$ with size exactly $k$, of which there are $\binom{f+k}{k} = \Omega_k (f^k)$.
    This collection has no FT-certificate other than itself; indeed, for any $P \in \binom{C}{k}$, failing the $f$ elements $F = C-P$ causes every $P' \in \binom{C}{k}$ with $P' \neq P$ to fail, so $P$ must be in the certificate.
\end{proof}

We now proceed to give an optimal (in terms of certificate size) strategy for Alice. 

\begin{claim}[Optimal Alice]\label{claim:optimal_alice}
    Alice has a strategy that guarantees $|\pp_i| \leq \binom{f+k}{k}$ for all $i$.
\end{claim}

\begin{proof}

Alice's strategy is to keep $P_i$ whenever it must be kept. 
She maintains $\pp_i$, as follows.
Initially, $\pp_0 = \emptyset$.
When Bob presents $P_i$, Alice checks if there exists $F_{i}\subseteq C$, such that:
\begin{enumerate}
    \item $|F_{i}| \leq f$,

    \item $F_{i}\cap P_{i}=\emptyset$, and

    \item For every $P\in \pp_{i-1}$, $P\cap F_i\neq \emptyset$.
\end{enumerate}
If there exists such $F_i$, then Alice keeps $P_i$, meaning $\pp_i \gets \pp_{i-1} \cup \{P_i\}$.
Otherwise, she discards $P_i$, meaning $\pp_i \gets \pp_{i-1}$.
It is easy to verify (by induction) that $\pp_i$ is an FT-certificate for $\pp_i^\all$.

We now proceed to the size analysis.
Consider the collection of pairs $X = \{ (P_j,F_j) \}_{P_j\in \pp_i}$.
By construction, we have:
\begin{enumerate}
    \item $P_j \cap F_j =\emptyset$. 

    \item $P_j \cap F_{j'} \neq \emptyset$ for $j<j'$. (when $F_{j'}$ was added, it hit all previous $P_j$)
\end{enumerate}
Such collection $X$ of pairs of $k$-sets and $f$-sets is called a skew cross-intersecting family, and it is known that such families consist of at most $\binom{f+k}{k}$ pairs~\cite{Bollobs1965, FRANKL1982125} (see \cite{Jukna11}, section 8.4 for discussion and extensions).
The claim follows.
\end{proof}

While \Cref{claim:optimal_alice} suggests an optimal strategy for Alice, it does not give her a \emph{practical} strategy. For each $P_i$, finding a blame set $F_i$, or determining no such $F_i$ exists, is an instance of the hitting set problem, which is NP-hard (when $f$ is part of the input).
On the other hand, we next show that by using a park, Alice can process and decide on each $P_i$ in $O_k(1)$ time, while still maintaining a certificate of size $O_k(f^k)$.
This highlights the usefulness of parks in fault-tolerant settings.
Another benefit is that the score function of a park gives a way to quantify the fault-tolerance with respect to color classes (links) in a linear way (as long as the park conditions are being met). These properties of parks are used extensively in \Cref{sect:analysis}.
We now turn to a formal description of this park-based strategy.

\begin{claim}[Efficient Alice]\label{claim: alice_uses_parks}
    Alice has a strategy that guarantees $|\pp_i| = O_k (f^k)$ for all $i$. Moreover, at round $i$ she decides whether to keep $P_i$ in $O_{k}(1)$ time.
\end{claim}
\begin{proof}
    Alice's strategy is ``dual'' to the optimal Alice of \Cref{claim:optimal_alice}. Instead of aiming to keep each $P_i$ that should be kept, she focuses on discarding $P_i$ in cases it can be safely discarded.
    
    Formally, the strategy of Alice is to keep the collection $\pp_i$ a park\footnote{Note that we defined the notion of a park for path collections, but in fact, it pertains only to the \emph{color sets} of the paths in the collection; i.e., we can think about collections of color sets instead of paths.}
    with respect to $\sco(\cdot) := \sco^{2,1/2}(\cdot)$ (i.e., $\alpha = 2$ and $\beta = 1/2$).
    Initially, $\pp_0 = \emptyset$.
    When Bob presents $P_i$, Alice checks if there is some $J \subseteq P_i$ such that $\pp_{i-1}$ is $J$-full:
    If there is, she discards $P_i$, meaning $\pp_i \gets \pp_{i-1}$.
    Otherwise, she keeps $P_i$, meaning $\pp_i \gets \pp_{i-1} \cup \{P_i\}$.

    We claim (by induction) that $\pp_i$ indeed remains a park when Alice keeps $P_i$ (if she discards $P_i$ this is trivial):
    For $J \not\subseteq P_i$, $\sco_J (\pp_i) = \sco_J (\pp_{i-1}) \leq 1$,
    and for $J \subseteq P_i$,
    \[
    \sco_J (\pp_i) = \sco_J (\pp_{i-1}) + \sco_J (P_i) \leq \frac{1}{2} + \frac{1}{2} \cdot (2f)^{-|P_i - J|} \leq 1~.
    \]

    We show (again by induction) that $\pp_i$ is an FT-certificate for $\pp^{\all}_i$.
    Let $F \subseteq C$ with $|F|\leq f$ be a faulty set, and suppose $P \in \pp^{\all}_i$ is non-faulty.
    If $P \in \pp_{i-1}$, then there is some non-faulty $P' \in \pp_{i-1} \subseteq \pp_i$.
    In the complementary case, we have $P = P_i$, so if Alice kept $P_i$ we are done.
    Otherwise, there is some $J \subseteq P_i$ (and hence, $J \cap F = \emptyset)$ such that $\pp_{i-1}$ is $J$-full, i.e., $\sco_J (\pp_{i-1}) > 1/2 = 1/\alpha$.
    So, by fault-tolerance of parks (\Cref{lemma:union_bound_scores}), there is a non-faulty $P' \in \pp_{i-1} [J] \subseteq \pp_i$.
    
    Next, we bound the size of $\pp_i$: Because $|P|\leq k$ for each $P \in \pp_i$, and $\pp_i$ is a park, we obtain
    $
    1 \geq \sco(\pp_i) \geq |\pp_i| \cdot \frac{1}{2} (2f)^{-k}
    $,
    hence $|\pp_i| = O_k (f^k)$.

    Finally, we observe that Alice can implement the strategy above efficiently. Specifically, she can maintain a dictionary whose keys are all the sets $J$ such that $\pp_i[J] \neq \emptyset$, with associated value $\sco_J(\pp_i)$ for each such $J$.
    Using such a dictionary, she can implement her strategy (by accessing the dictionary in $O_k(1)$ locations at each round and checking if the relevant links are full). Moreover, once Alice inserted $P_i$ she needs to update the link of each subset of $P_i$, which also takes $O_k(1)$. 
\end{proof}

%% file: parks-toolbox.tex
\section{Parks Toolbox}\label{sect:parks-toolbox}

\subsection{Concatenating an Edge}

In this section, we prove \Cref{clm:concat-edge}.
\concatedge*
Moreover, we have an analogous statement for vertex colors.
\begin{claim}[Concatenating an Edge, Vertex Colors]\label{claim:concat_edge_vertex_colors}
    Consider the same setting as \Cref{clm:concat-edge} but with vertex colors, i.e., $c:V\to C$.
    Suppose $e=\{v,u\}$ and that $\pp$ is stemming from $u$.
    Then:
    \begin{enumerate}[label=(\arabic*)]
        \item $\sco_J(\pp) \leq \sco_{J\cup \{c(v)\}}(e \circ \pp)$.

        \item $\sco_J(e \circ \pp)\leq \sco_J(\pp) + \sco_{J-\{c(v)\}}(\pp)$. 
    \end{enumerate}
\end{claim}
We now prove \Cref{clm:concat-edge}. The proof of the analogous claim for vertex colors is identical, and is by just replacing every appearance of $c(e)$ with $c(v)$.
\begin{proof}[Proof of \Cref{clm:concat-edge}.]
    For (1), observe that $|c(e\circ P) - (J \cup \{c(e)\})| \leq |c(P) - J|$ for every $P \in \pp$, hence
    \begin{align*}
        \sco_J(\pp) 
        &= \sum_{P\in \pp[J]} \sco_J(P) \leq \sum_{P\in \pp[J]} \sco_{J \cup \{c(e)\}}(e\circ P)
        =\sum_{P'\in e\circ(\pp[J])} \sco_{J \cup \{c(e)\}}(P'), \\
        \shortintertext{which is, since $e\circ (\pp[J]) = (e \circ \pp)[J \cup \{c(e)\}]$,}
        &= \sum_{P'\in (e\circ\pp)[J\cup \{c(e)\}]} \sco_{J \cup \{c(e)\}}(P')
        = \sco_{J \cup \{c(e)\}}(e\circ \pp).
    \end{align*}

    For (2), we split to cases.
    The easy case is when $c(e) \notin J$.
    Then, we have $(e\circ \pp)[J] = e \circ (\pp[J])$. As adding colors to a path only reduces its score, we obtain that $\sco_J(e \circ \pp)\leq \sco_J(\pp)$, which is stronger than what we needed to prove.
    We now focus on the case when $c(e) \in J$.
    Denote $\pp_1 = \pp[J - \{c(e)\}] - \pp[J]$, and $\pp_2 = \pp[J]$.
    Then
    \begin{align*}
        \sco_J (e \circ \pp_1)
        =\sum_{P\in \pp_1} \beta (\alpha f)^{-|c(e\circ P) - J|}
        = \sum_{P \in \pp_1} \beta (\alpha f)^{-|c(P) - (J-\{c(e)\})|}
        = \sco_{J-\{c(e)\}} (\pp_1) ,
    \end{align*}
    and
    \begin{align*}
        \sco_J (e \circ \pp_2)
        =\sum_{P\in \pp_2} \beta (\alpha f)^{-|c(e\circ P) - J|}
        =\sum_{P\in \pp_2} \beta (\alpha f)^{-|c(P) - J|}
        =\sco_J (\pp_2).
    \end{align*}
    Note that $(e \circ \pp)[J] = e\circ \pp[J-\{c(e)\}] = (e \circ \pp_1) \cup (e \circ \pp_2)$.
    Hence,
    \begin{align*}
        \sco_J (e \circ \pp) 
        &= \sco_J (e \circ \pp_1) + \sco_J (e \circ \pp_2) \\
        &= \sco_{J-\{c(e)\}} (\pp_1) + \sco_J (\pp_2)
        \leq  \sco_{J-\{c(e)\}} (\pp) + \sco_J (\pp) ,
    \end{align*}
    as needed.
\end{proof}

\subsection{Linkful Maps}\label{sec : linkful}

We recall the definition of linkful maps, and the main lemma concerning them, whose proof is the purpose of this section.

\deflinkful*
\lemlinkful*

Before proving \Cref{lem:linkful_map_new}, we need two easy observations regarding scores and parks:

\begin{observation}[Score Transition]\label{lem:score_transition}
    For every $(\alpha,\beta)$ score function $\sco(\cdot)$, path collection $\pp$ and
    $X,Y\subseteq C$, we have
    \[
    \sco_X(\pp [Y]) = \sco_Y(\pp [X\cup Y]) (\alpha f)^{|X|-|Y|}.
    \]
\end{observation}
\begin{proof}
    By \Cref{def:score} (score), 
    \[
    \sco_X(\pp [Y]) 
    = \sco_X((\pp [Y])[X]) 
    = \sco_X(\pp [X \cup Y])
    \]
    Hence, we show:
    \begin{equation}\label{eq:relative_score_transition}
    \sco_X(\pp [X \cup Y]) = \sco_Y(\pp [X\cup Y]) (\alpha f)^{|X|-|Y|}.
    \end{equation}
    By linearity of $\sco(\cdot)$ on path collections, it is enough to prove \Cref{eq:relative_score_transition} for $\pp=\{P\}$ for a path $P$. 
    If $X\cup Y \not \subseteq c(P)$, then $\pp[X \cup Y] = \emptyset$ and both sides of \Cref{eq:relative_score_transition} are $0$.
    Otherwise, $X\cup Y  \subseteq c(P)$, thus $\pp[X \cup Y] = \pp$. 
    By \Cref{def:score} (score), $\sco_X(P) = \beta (\alpha f)^{|X|-|c(P)|}$, and $\sco_Y(P) = \beta (\alpha f)^{|Y|-|c(P)|}$. 
    Plugging this into \Cref{eq:relative_score_transition} implies the observation.
\end{proof}

\begin{observation}\label{obs:double_counting_all_paths}
    Let $\pp$ be a park w.r.t.\ an $(\alpha,\beta)$-score function $\sco(\cdot)$.
    If every path in $\pp$ is of length at most $i$, then
    \[\sco(\pp) \geq \frac{1}{2^{i}} \sum_{J \subseteq C} \sco(\pp[J]).\]
\end{observation}

\begin{proof}
    By linearity of $\sco$ on path collection, it is enough to prove this for a single path $P$, so suppose $\pp=\{P\}$.
    By \Cref{def:score},
    $\sco(\pp[J]) = \sco(P)$ if and only if $J \subseteq c(P)$, thus there are at most $2^{|c(P)|}\leq 2^i$ sets $J\subseteq C$ that contribute to the RHS.
\end{proof}

We are now ready to prove the Linkful Map Lemma (\Cref{lem:linkful_map_new}):

\begin{proof}[Proof of \Cref{lem:linkful_map_new}]
    Without loss of generality, we may assume that $I \subseteq c(P)$ for every $P \in \pp$.
    Otherwise, we can replace $\pp$ with $\pp[I]$, apply the lemma with the appropriate restriction of the linkful map, and use the fact that $\sco_I (\pp) = \sco_I (\pp[I])$.

    For each $T \subseteq I$, let
    $\pp_T$ be the collection of all paths $P \in \pp$ such that $g(P) \cap I = T$.
    Then $\{\pp_T\}_{T \subseteq I}$ are mutually disjoint and their union is $\pp$.
    Hence,
    $
    \sco_I (\pp) = \sum_{T \subseteq I} \sco_I (\pp_T)
    $,
    so there must be some $T \subseteq I$ such that
    \begin{equation}\label{eq:T-set-ineq}
        \sco_I (\pp_T) \geq \tfrac{1}{2^{|I|}}\sco_I (\pp).
    \end{equation}
    We now fix $T$ to be a set satisfying \Cref{eq:T-set-ineq}, and prove $T$ satisfies \Cref{eq:linkful}.

    First, if $T = g(P)$ for some $P \in \pp$, then the link of $T$ must be full in $\pp'$ (because $g$ is a linkful map), meaning that $\sco_T (\pp') > 1/2$, and we are done.
    So, we assume henceforth that $g(P) \neq T$ for all $P \in \pp$.
    Let us first focus on some arbitrary set $X \in g(\pp_T)$, meaning that $X = g(P)$ for some $P \in \pp_T$, and thus the following properties holds:
    \begin{enumerate}
        \item[(1)] $X \neq T$, by our current assumption.
        \item[(2)] $X \cap I = T$, by definition of $\pp_T$. In particular, $T \subseteq X$.
        \item[(3)] $\pp'$ is $X$-full, since $g$ is a linkful map.
    \end{enumerate}
    We have that 
    \begin{align*}
        \sco'_T (\pp'[X])
        &= \frac{\sco'_T (\pp'[X])}{\sco_I (\pp_T [X])} \cdot \sco_I (\pp_T [X])
    \shortintertext{by \Cref{lem:score_transition},}
        &= \frac{(\alpha' f)^{|T|-|X|} \sco'_X (\pp'[X \cup T])}{(\alpha f)^{|I| - |I \cup X|} \sco_{I\cup X} (\pp_T [I \cup X])} \cdot \sco_I (\pp_T [X])
    \shortintertext{by property (2), we have $|X|-|T| = |I \cup X| - |I|$, hence}
        &= (\frac{\alpha}{\alpha'})^{|X|-|T|} \cdot 
        \frac{\sco'_X (\pp'[X \cup T])}{\sco_{I\cup X} (\pp_T [I \cup X])} \cdot \sco_I (\pp_T [X])
    \shortintertext{since $\alpha \geq \alpha'$, and $|X|-|T| \geq 1$ by properties (1) and (2),}
        &\geq \frac{\alpha}{\alpha'} \cdot 
        \frac{\sco'_X (\pp'[X \cup T])}{\sco_{I\cup X} (\pp_T [I \cup X])} \cdot \sco_I (\pp_T [X])
    \end{align*}
    By property (2), the numerator in the middle term is just $\sco_X (\pp')$, which is more than $1/2$ by property (3).
    Also, the denominator in this term is at most $\sco_{I\cup X} (\pp)$, which is at most $1$ since $\pp$ is a park.
    In conclusion, we obtained:
    \begin{equation}\label{eq:X-set-ineq}
        \sco'_T (\pp'[X]) > \frac{\alpha}{2\alpha'} \sco_I (\pp_T [X]), \quad \text{for every $X \in g(\pp_T)$.}
    \end{equation}
    Now, using \Cref{obs:double_counting_all_paths} and the assumption that paths in $\pp'$ have at most $\ell$ colors, we get
    \begin{align*}
        \sco'_T (\pp') 
        &\geq \frac{1}{2^\ell} \sum_{X \in g(\pp_T)} \sco'_T (\pp'[X])
        \shortintertext{by \Cref{eq:X-set-ineq},}
        &> \frac{\alpha}{2^{\ell+1} \alpha'} \sum_{X \in g(\pp_T)} \sco_I (\pp_T[X])
        \shortintertext{as each $P \in \pp_T$ belongs to $\pp[g(P)]$,}
        &\geq \frac{\alpha}{2^{\ell+1} \alpha'} \sco_I (\pp_T)
        \shortintertext{since we chose $T$ that satisfies \Cref{eq:T-set-ineq},}
        &\geq \frac{\alpha}{2^{\ell+|I|+1} \alpha'} \sco_I (\pp)
    \end{align*}
    as required.
\end{proof}

%% file: sampling.tex
\section{The Park Sampling Algorithm}\label{sect:sampling}

This section is devoted to the proof of \Cref{lem:sampling_new}:
\lemsampling*

Throughout, we assume that $\prepp = \prepp [I]$ (as we can ignore the paths outside the $I$-link).
Also, we slightly abuse notations, and even though the parks in this section are \emph{ending at} $S_i$, we use the notation for parks \emph{stemming from} $S_i$, so $\prepp_s$  consists of all the paths from $\prepp$ that end in $s \in S_i$.

Intuitively, the approach is to take a batch of roughly $O( 1/p \cdot \log n)$ vertices in $S_i$, 
with the following property:
Every vertex $s$ in this batch have roughly the same score $\prelsco_I^i(\prepp_{s})$.
W.h.p., at least one vertex $s^*$ from this batch is sampled to $S_{i+1}$,
and we add the paths in $\prepp_{s^*}$ to $\pp$. 
Every other vertex $s$ in the batch is discarded, and the paths of $\prepp_{s}$ are removed from $\prepp$ (these removed paths are denoted $\prepp^{\buc}$).
Once $\pp$ is $J$-full for some $J\subseteq C$, we remove all the paths of $\prepp[J]$ from $\prepp$ (these removed paths are denoted $\prepp^{\col}$).
We repeat this procedure until $\pp$ is $I$-full.

We formalize the above in \Cref{alg:touristic_sampling}, where we use the following bucketing approach to yield the batches.
For $\prepp'  \subseteq \prepp$ and integer $j \geq 1$ we denote
\[
B_j(\prepp') = \{s \in S_i \mid 2^{-j} <  \gsco^{i+1}_{I}(\prepp'_{s}) \leq 2^{-j+1} \}
\]
We say $B_j(\prepp')$ is \emph{samplable} if $|B_j(\prepp')| \geq 1/\rho$.

\begin{algorithm}[H]
    \caption{Park Sampling}\label{alg:touristic_sampling}
    \begin{algorithmic}[1]
    \State $\pp \gets \emptyset$
    \State $\prepp^{'} \gets \prepp$

    \While{$\exists j$ such that $|B_{j}(\prepp^{'})| \geq 1/\rho$}
        \State pick $S \subseteq B_{j}(\prepp^{'})$ of size $1/\rho$
        \If{$S \cap S_{i+1} = \emptyset$}
            \Return ``Error'' \Comment{no $(i+1)$-center in $S$}
        \EndIf
        \State pick $s^* \in S \cap S_{i+1}$
        \State  $\pp \gets \pp \cup \prepp^{'}_{s^*}$
        \State $\prepp^{'} \gets \prepp^{'} - \bigcup_{s\in S } \prepp^{'}_{s} $ \label{line:del_by_buc}
        \For{each $J$ s.t. $\pp$ is $J$-full (w.r.t.\ $\gsco^{i+1}$) and  $\prepp^{'}[J]$ is non empty} \label{line:for_loop_del_by_col}
            \State $\prepp^{'} \gets \prepp^{'} - \prepp^{'}[J]$ \label{line:del_by_col}
        \EndFor
    \EndWhile
    \State \Return $\pp$
    \end{algorithmic}
\end{algorithm}

Denote by $\pp^{\buc}$ the set of paths deleted in \Cref{line:del_by_buc} of \Cref{alg:touristic_sampling}, and by $\pp^{\col}_v$ the set of paths deleted in \Cref{line:del_by_col}.
The algorithm maintains the following invariants:
\begin{itemize}
\item[]
\begin{enumerate}[label={({\bfseries S\arabic*})}]
    \item $\prepp$ is the disjoint union of $\prepp'$, $\prepp^{\buc}$ and $\prepp^{\col}$. \label{inv:S1}
    \item $\pp$ is a $(\gsco^{i+1},\lsco^{i+1})$-touristic park. \label{inv:S2}
\end{enumerate}
\end{itemize}
Invariant~\ref{inv:S1} clearly holds.
For Invariant~\ref{inv:S2}, note that it holds at initialization, and the following claim shows that it continues to hold:

\begin{claim}\label{clm:no-overshooting-sampling}
    Suppose that at the beginning of an iteration of the while loop in \Cref{alg:touristic_sampling}, $\pp$ is a $(\gsco^{i+1},\lsco^{i+1})$-touristic park.
    Then $\pp \cup \prepp^{'}_{s^*}$ is also  a $(\gsco^{i+1},\lsco^{i+1})$-touristic park.
\end{claim}
\begin{proof}
    The local condition is immediate, since $\pp \cup \prepp'_{s^*} \subseteq \prepp$, and the latter is touristic w.r.t.\ the local score function $\prelsco^i$, 
    which always bounds from above
    $\lsco^{i+1}$ by our choice of parameters.

    For the global condition, let $J \subseteq C$, and we show that $\gsco^{i+1}_J(\pp \cup \prepp^{'}_{s^*}) \leq 1$.
    We split to cases:
    \begin{itemize}
        \item $\pp$ is $J$-full: 
        Then it must be that $\prepp'[J] = \emptyset$, because of \Cref{line:for_loop_del_by_col,line:del_by_col} executed at the end of the previous while-loop iteration.
        Thus, $\gsco^{i+1}_J(\pp \cup \prepp^{'}_{s^*}) = \gsco^{i+1}_J(\pp) \leq 1$.
        
        \item $\pp$ is not $J$-full: Then $\gsco^{i+1}_J (\pp) \leq 1/2$, so it's enough to prove that $\gsco^{i+1}_J (\prepp^{'}_{s^*}) \leq 1/2$.
        But this is easy:
        \begin{align*}
            \gsco^{i+1}_J (\prepp^{'}_{s^*})
            & \leq \frac{1}{D} \prelsco^i_J (\prepp^{'}_{s^*}) && \text{as $\gsco^{i+1} \leq \frac{1}{D} \prelsco^i$ by choice of parameters,} \\
            & \leq \frac{1}{D} && \text{as $\prepp^{'}_{s^*} \subseteq \prepp_{s^*}$, which is a park w.r.t. $\prelsco^i$} \\
            & \leq \frac{1}{2} && \text{as $D \geq 2$.}
        \end{align*}
    \end{itemize}
\end{proof}

We now build towards proving that the \Cref{alg:touristic_sampling} is (w.h.p) correct.

\begin{claim}\label{claim:exists_samplable_bucket}
    If at the end of an iteration, $\pregsco_I^i(\prepp')> 1/8$, then exists $j$ s.t. $|B_j(\prepp')| \geq 1/\rho$.
\end{claim}
\begin{proof}
    We will prove the contra-positive.
    Suppose for every $j$, $|B_j(\pp')| < 1/\rho$, then we will show that $\pregsco^{i}_{I}(\prepp') \leq 1/8$.
    We now analyze,
    \begin{align*}
        \pregsco^{i}_{I}(\prepp') 
        &= \sum_{j\geq 1} \sum_{s \in B_j(\prepp')} \pregsco^{i}_{I} (\prepp'_{s})  \\
        &\leq \sum_{j\geq 1} \sum_{s \in B_j(\prepp')}  \frac{\hat{\beta}_{i}\rho}{\beta_{i+1}} \gsco^{i+1}_{I} (\prepp'_{s})
        && \text{as $\alpha_{i+1}\leq \hat{\alpha}_i$} \\
        &\leq 
        \sum_{j\geq 1} |B_j(\prepp')| \cdot \frac{\hat{\beta}_{i}\rho}{\beta_{i+1}} \cdot \frac{1}{2^{j-1}} && \text{as $\forall s\in B_j(\prepp'), \ \gsco^{i+1}_{I} (\prepp'_{s})\leq \tfrac{1}{2^{j-1}}$} \\
        &\leq \sum_{j\geq 1} \frac{1}{D^3} \cdot \frac{1}{2^{j-1}} && \text{as $\forall j, \ |B_j(\prepp')|\leq \tfrac{1}{\rho}$ and $\hat{\beta}_{i} = \tfrac{\beta_{i+1}}{D^3}$} \\
        &\leq \frac{1}{D^3}\cdot 2 \leq \frac{1}{8}.
    \end{align*}
\end{proof}

\begin{claim}\label{claim:del_by_buc}
    Before (and after) each iteration of the while loop in \Cref{alg:touristic_sampling}, it holds that
    \[
    \pregsco_I^i(\prepp^\buc) \leq \frac{2}{D} \gsco_I^{i+1}(\pp) \leq \frac{1}{8}.
    \]
\end{claim}
\begin{proof}
    By induction.
    Before the first iteration, both scores are $0$.
    Now, assume that the claim holds at the beginning of some iteration.
    Then
    \begin{align*}
        \pregsco_I^i \left( \prepp^\buc \cup \bigcup_{s \in S} \prepp'_s \right)
        &\leq  \pregsco_I^i(\prepp^\buc) + \sum_{s \in S} \pregsco_I^i (\prepp'_s) && \text{by union bound,} \\
        &\leq \frac{2}{D} \gsco_I^{i+1}(\pp) + \sum_{s \in S} \pregsco_I^i (\prepp'_s) && \text{by induction,} \\
        &\leq \frac{2}{D} \gsco_I^{i+1}(\pp) + \sum_{s \in S} \frac{\hat{\beta}_i \rho}{\beta_{i+1}} \gsco_I^{i+1} (\prepp'_s) &&\text{as $\alpha_{i+1} \leq \hat{\alpha}_i$,} \\
        &\leq \frac{2}{D} \gsco_I^{i+1}(\pp) + \frac{1}{\rho} \cdot \frac{\hat{\beta}_i \rho}{\beta_{i+1}} \cdot 2 \cdot  \gsco_I^{i+1}(\prepp'_{s^*}) &&\text{as $s^* \in S \subseteq B_j(\prepp')$, $|S| = 1/\rho$,}\\
        &\leq \frac{2}{D} \gsco_I^{i+1}(\pp) + \frac{2}{D} \gsco_I^{i+1}(\prepp'_{s^*}) &&\text{by choice of parameters,} \\
        &= \frac{2}{D} \gsco_I^{i+1} (\pp \cup \prepp'_{s^*}) \\
        &\leq \frac{1}{8},
    \end{align*}
    as $\pp \cup \prepp'_{s^*}$ is a park by \Cref{clm:no-overshooting-sampling}, and $D \geq 16$.
    This concludes the induction step.
\end{proof}

\begin{claim}\label{claim:del_by_col}
    If at the end of an iteration, $\pregsco_I^i(\prepp^\col) \geq 1/8$, then $\pp$ is $I$-full.
\end{claim}
\begin{proof}
    Suppose $\pregsco_I^i(\prepp^\col) \geq 1/8$.
    For every $P\in \prepp^{\col}$, denote by $J(P)$ 
    the set of colors that caused $P$ to be added to 
    $\prepp^{\col}$, i.e., $\pp$ is $J(P)$-full and thus $J$ is a linkful map between $\prepp^{\col}$ and $\pp$.
    Hence, by \Cref{lem:linkful_map_new}, exists $T\subseteq I$ such that 
    \[
    \gsco^{i+1}_T (\pp) 
    > \min \left\{\frac{1}{2}, \frac{\hat{\alpha}_i}{2^{i+2+|I|}\alpha_{i+1}}\pregsco^i_I (\prepp^\col)\right\}
    \geq \min \left\{\frac{1}{2}, \frac{\hat{\alpha}_i}{2^{i+6}\alpha_{i+1}}\right\}
    \geq \frac{1}{2},
    \]
    as by our choice of parameters, 
    $\frac{\hat{\alpha}_i}{2^{i+6}\alpha_{i+1}} 
    = \frac{D^{10k}}{2^{i+6}}>\tfrac{1}{2}$. 
    Since $\pp = \pp[I]$ and $T \subseteq I$, by definition of score (\Cref{def:score}) we have $\gsco_I^{i+1}(\pp)\geq \gsco_T^{i+1}(\pp)$, which concludes the proof.
\end{proof}

We are finally ready to prove the correctness of \Cref{alg:touristic_sampling}.

\begin{claim}\label{claim:error-prob}
    W.h.p., \Cref{alg:touristic_sampling} returns an $I$-full park.
\end{claim}
\begin{proof}
    We first show that w.h.p., \Cref{alg:touristic_sampling} does not return ``Error''.
    
    We begin with analyzing the probability to return ``Error'' in a single iteration of the while loop. 
    Suppose the algorithm did not return ``Error'' until the current iteration, and consider the current set $S$.
    Fixing the randomness of previously considered centers (essentially, the endpoints of all the paths in $\prepp^\buc$), the current set $S$ is fixed.
    By \Cref{line:del_by_buc}, this current $S$ is disjoint from the previously considered centers, and since every $s\in S_i$ is sampled independently to be in $S_{i+1}$, it still holds that every $s\in S$ is sampled independently (from the past, and from the other elements in $S$) with the same probability $p$ to be in $S_{i+1}$.
    Thus, the probability that the algorithm errors in a single iteration is $(1-p)^{1/\rho}\leq e^{-p/\rho}$.

    We now bound the number of iterations.
    Every path added to $\pp$ increases $\gsco^{i+1}_{I}(\pp)$ by at least $\beta_{i+1}(\alpha_{i+1}f)^{-i-1}$.
    By Invariant~\ref{inv:S2},
    $\pp$ remains a park w.r.t.\ $\gsco^{i+1}$ throughout the execution of \Cref{alg:touristic_sampling}, so its $I$-score cannot exceed $1$.
    Hence, there can be at most $(\beta_{i+1})^{-1}(\alpha_{i+1}f)^{i+1}\leq 2^{O(k^3)} f^k$ 
    iterations in the while loop of \Cref{alg:touristic_sampling}.
 
    By a union bound and the choice of $\rho$ (see \Cref{eq:rho-def}), the probability that \Cref{alg:touristic_sampling} returns ``Error" is at most
    $e^{-\frac{p}{\rho}} \cdot 2^{O(k^3)} f^k \leq 1/\poly(n)$, where the degree of the polynomial can be set to an arbitrarily large constant by appropriately setting the constants in the definition of $\rho$. 
    
\medskip
    
    Assume now that the algorithm terminates without returning ``Error'', and consider the state of all variables after halting.
    By Invariant~\ref{inv:S1}, and using that $\prepp$ is $I$-full, we obtain
    \[
    \pregsco_I^i(\prepp') + \pregsco_I^i(\prepp^\buc) + \pregsco_I^i(\prepp^\col) = \pregsco_I^i(\prepp)> \frac{1}{2}.
    \]
    By \Cref{claim:exists_samplable_bucket}, $\pregsco_I^i(\prepp') \leq 1/8$ and
    by \Cref{claim:del_by_buc}, $\pregsco_I^i(\prepp^\buc)\leq 1/8$.
    Hence,  it must be that $\pregsco_I^i(\pp^\col)> 1/8$.
    Thus, we can apply \Cref{claim:del_by_col} and conclude the proof.
\end{proof}

\paragraph{Running Time of the Sampling Algorithm.}
By using standard data structures, we may assume that $\prepp$ has an associated data structure that supports the following operations:
(1) given $s \in S_i$, return $\prepp_s$ and $\pregsco_I(\prepp_s)$ in $\tilde{O}_k(1)$ time,
(2) given $s \in S_i$ with a command ``Delete'', remove $\prepp_s$ from $\prepp$ in $\tilde{O}_k(|\prepp_s|)$ time,
(3) given $J\subseteq C$ with a command ``Delete'', it removes $\prepp[J]$ from $\prepp$ in $\tilde{O}_k(|\prepp[J]|)$ time,
(4) given command ``Get Heavy Bucket'', return $1/\rho$ elements from a bucket $B_j (\prepp)$ (or ``Error'' if such does not exists), in $O(1/\rho)$ time.
Similarly, $\pp$ supports inserting $\prepp_s$ in the same running time and querying $\gsco^{i+1}_J(\pp)$ in $\tilde{O}(1)$ time.
Using these, all the deletion and insertion operations take a total of $\tilde{O}_k(|\prepp[I]|)$ time.
It suffices to check the condition of \Cref{line:for_loop_del_by_col} only for $J\subseteq c(P)$ where $P\in \prepp'_{s^*}$ (as otherwise, this condition is clearly not satisfied), which is again bounded by a total of $\tilde{O}_k(|\prepp|)$ time.
The overhead of computing the buckets is only a second order term, and the proof is concluded.

\begin{proof}[Proof of \Cref{lem:sampling_new}.]
    The Lemma follows from \Cref{claim:error-prob} and the running time analysis above.
\end{proof}